\PassOptionsToPackage{hyphens}{url}
\documentclass[11pt,letterpaper]{article}
\usepackage[utf8]{inputenc}
\usepackage[T1]{fontenc}
\usepackage{lmodern}
\usepackage[margin=1in]{geometry}
\usepackage{microtype}
\usepackage{amsmath,amssymb,amsthm,mathtools}
\usepackage{booktabs,eurosym,enumitem}
\usepackage{graphicx,xcolor,comment}
\usepackage[normalem]{ulem}
\usepackage{placeins,pdflscape,subcaption,array}
\usepackage{threeparttable,bbold,relsize,tikz,todonotes}
\usepackage[export]{adjustbox}
\usepackage[authoryear,round]{natbib}
\usepackage{authblk}
\usepackage[colorlinks=true,linkcolor=blue,urlcolor=blue,citecolor=blue]{hyperref}
\hypersetup{pdftitle={Identification-aware Markov Chain Monte Carlo},pdfauthor={Toru Kitagawa and Yizhou Kuang}}
\numberwithin{equation}{section}

\newcolumntype{L}[1]{>{\raggedright\let\newline\\\arraybackslash\hspace{0pt}}m{#1}}
\newcolumntype{C}[1]{>{\centering\let\newline\\\arraybackslash\hspace{0pt}}m{#1}}
\newcolumntype{R}[1]{>{\raggedleft\let\newline\\\arraybackslash\hspace{0pt}}m{#1}}
\theoremstyle{plain}

\newtheorem{lemma}{Lemma}
\newtheorem{proposition}{Proposition}
\newtheorem{corollary}{Corollary}
\newtheorem{assumption}{Assumption}
\newtheorem{algorithm}{Algorithm}
\theoremstyle{definition}

\newtheorem{example}{Example}
\newtheorem{remark}{Remark}
\newtheorem{simulation}{Simulation}
\title{Identification-aware Markov Chain Monte Carlo\thanks{We would like to thank Donald Andrews, Karim Chalak, Ben Deaner, Raffaella Giacomini, Frank Kleibergen, Yuan Liao, José Luis Montiel Olea, Kristoffer Nimark, Andriy Norets, Alexei Onatski, Mikkel Plagborg-Møller, Anna Simoni, Tomasz Woźniak, and Andrei Zeleneev for useful discussions. We also thank seminar participants at the University of Cambridge, Singapore Management University and conference participants at CFE-CMStatistics, Aarhus Econometrics Workshop, NASMES 2026, SETA 2026, and ISBA 2026 for their valuable feedback. All remaining errors are our own.}}
\author[1]{Toru Kitagawa}
\author[2]{Yizhou Kuang}
\affil[1]{Brown University}
\affil[2]{University of Manchester}
\date{This version: October 4, 2026}
\begin{document}
\maketitle
\begin{abstract}
In Bayesian analysis, a proper posterior can exist even when identification fails, but its existence does not ensure reliable posterior computation. More broadly, posterior simulation can be slow when local samplers must move between separated regions or along ridges and weakly curved directions. We develop identification-aware Markov chain Monte Carlo (MCMC) samplers that combine local updates with moves across observationally equivalent parameter values, thereby giving identification analysis a constructive role in modern Bayesian analysis. Under uniform mixing and communication conditions, we establish spectral-gap comparisons in which local-sampler upper bounds vanish as modes separate or likelihood-flat fibers expand, while lower bounds for identification-aware samplers remain bounded away from zero. Simulations illustrate these mechanisms. Structural vector moving-average (SVMA) applications at two scales demonstrate substantial sampling improvements under weak priors, with limited additional gains under
strong priors.

\end{abstract}
\clearpage
\section{Introduction} 
Lindley's assertion that ``unidentifiability causes no real difficulty in the Bayesian approach'' \citep[p.~46]{lindley1972bayesian} reflects a widely held Bayesian view, echoed by \citet{gelman2014identifiability}: lack of identification need not prevent inference when the posterior is proper. However, forming a posterior and computing it reliably are distinct tasks.

In practice, researchers often use Markov chain methods\footnote{Examples include Metropolis--Hastings (MH), Hamiltonian Monte Carlo (HMC) and its No-U-Turn Sampler (NUTS) variant, and sequential Monte Carlo (SMC). Although SMC is not an MCMC method, implementations such as \citet{herbst2014sequential} use local MCMC mutation steps.} with the hope that the simulated draws accurately represent the target posterior. These methods can be highly effective, but multimodality,
extended flat regions, and strongly curved or poorly scaled posterior geometry can impede exploration. A sampler may then remain confined to part of the posterior for a long time, even when convergence holds asymptotically. The resulting draws can give misleading estimates of posterior probabilities and uncertainty.

This paper establishes a constructive computational role for identification analysis in Bayesian inference, extending its relevance beyond the question of whether identification is required for posterior formation. Model-implied observational equivalence between parameter values provides additional information for constructing nonlocal moves, complementing the likelihood evaluations, gradients, and posterior features learned during simulation used by conventional methods.

Our identification-aware (IA) MCMC sampler combines a local kernel \(P\) with a teleport kernel \(T\) operating within observational equivalence (OE) sets:
\[
\widetilde P_H=(1-\varepsilon)P+\varepsilon T,
\qquad \varepsilon\in(0,1).
\]
These moves preserve the target posterior while creating opportunities to escape local traps and reach regions that standard updates struggle to explore. They can be incorporated into different baseline samplers, turning identification analysis into a practical tool for  improving posterior computation.

We establish theoretical guarantees for the proposed samplers in two
settings characterized by common difficulties in posterior simulation:
well-separated modes and extended likelihood-flat regions. Under certain regularity conditions, the corresponding local-sampler spectral-gap upper bounds vanish as modes separate or flat regions expand, whereas the IA samplers retain a common positive lower bound. The main results concern random-walk Metropolis. The Online Appendix gives separated-region extensions for fixed-length HMC and fixed-tuning reversible NUTS under additional conditions.

The numerical exercises demonstrate the computational gains from exploiting observational equivalence. Gaussian-mixture and mean-decomposition designs illustrate movement between equivalent modes and along likelihood-flat fibers. The SVMA applications examine the approach at two scales, from an SVMA(1) illustration to the 153-parameter news-shock specification of \citet{plagborg2019bayesian}. Under the bounded-uniform prior in the larger application, IA--NUTS reduces the mean scaled coefficient discrepancy from its numerical conditional-prior reference by approximately 74 percent relative to NUTS, while teleport updates account for only 2.1 percent of transition time. Under the informative prior, the two samplers give similar summaries. Across both scales, the results show substantial gains under weak priors, with limited additional gains under strong priors.

\paragraph{Related Literature.}
This paper contributes to the econometric literature on simulation-based Bayesian inference. MCMC and related simulation methods are now standard tools for Bayesian inference across a wide range of economic models \citep{chib1996markov, an2007bayesian, flury2011bayesian, herbst2014sequential, plagborg2019bayesian, meager2022aggregating}. Simulation methods are also used outside the Bayesian posterior-sampling framework: \citet{chernozhukov2003mcmc} use MCMC to compute extremum estimators, while \citet{chen2018monte} use SMC-type procedures to construct confidence sets for partially identified models, building on related sequential simulation methods such as \citet{herbst2014sequential} and \citet{durham2014adaptive}. Our contribution is a new approach to posterior simulation in which model-implied observational equivalence can be built directly into the sampling design.

A large theoretical literature studies convergence of MCMC and particle methods. Classical results for random-walk Metropolis and general-state-space Markov chains include \citet{mengersen1996rates}, \citet{roberts1996geometric}, \citet{roberts2004general}, and \citet{meyn2009markov}. For SMC and HMC, formal convergence, stability, and high-dimensional scaling results are given, among others, by \citet{chopin2004central}, \citet{del2012adaptive}, \citet{huggins2015convergence}, \citet{durmus2020irreducibility}, and \citet{mangoubi2021mixing}. Our analysis asks how exploiting observational equivalence changes the convergence of a given local sampler. It identifies conditions under which model-implied teleport moves prevent the deterioration caused by increasing mode separation or expanding likelihood-flat regions.

The paper is also related to methods that improve exploration of multimodal or irregular posteriors through auxiliary distributions, nonlocal moves, or adaptation to features learned during simulation. Tempering and equi-energy methods use auxiliary distributions or empirical energy levels \citep{kou2006equi,woodard2009sufficient}. The dynamic striated Metropolis--Hastings sampler of \citet{waggoner2016striated} uses previous-stage draws to construct posterior-height strata, while \citet{zhou2011multi} and \citet{pompe2020framework} identify modes or mode-associated regions and construct or adapt jumps among them. Our IA moves draw on an additional source of information:
observational equivalence implied by the model. This structure identifies
likelihood-equivalent candidates directly, complementing posterior features
learned during simulation and potentially reducing the computational effort
needed to discover useful nonlocal moves.

Several strands of the literature already exploit known equivalences. Permutation and augmentation methods address label-switching multimodality in mixtures \citep{celeux2000computational,fruhwirth2006finite, geweke2007interpretation,jasra2005markov,marin2005bayesian}, while rotation draws in sign-restricted SVARs exploit observational equivalence in a particular structural application \citep{uhlig2005effects,rubio2010structural}; see \citet{baumeister2015sign} for the associated prior implications. Related work on parameter-expanded data augmentation and orbital chains provides theory for particular classes of augmentation or symmetry kernels \citep{hobert2008theoretical,niepert2012orbits}, but leaves open how observational equivalence in structural models can be systematically exploited to overcome barriers to posterior exploration.

Finally, the proposed sampler is complementary to Bayesian and quasi-Bayesian work on partially identified models. \citet{kline2016bayesian} and \citet{giacomini2021robust} study Bayesian inference and sensitivity analysis for identified sets, while \citet{chen2018monte} develop frequentist-valid Monte Carlo confidence sets for partially identified models. Related robustness and sensitivity ideas also appear in \citet{kuang2026robust}. These papers focus primarily on inference under non-identification. We focus on computation: when observational equivalence or likelihood-flat structure is known, it can be used to construct exact posterior-preserving moves that improve exploration of the posterior or quasi-posterior.

Our contribution is a general formulation that uses identification analysis both to construct nonlocal moves and to analyze their effectiveness. It extends the computational use of observational equivalence beyond particular applications and connects the construction of these moves to an analysis of the mixing difficulties they can overcome. More broadly, the paper establishes a computational rationale for giving identification greater attention in Bayesian analysis, even when a proper posterior exists in a model that is not point-identified.


\section{MCMC and Identification-Aware Sampling}
\label{sec: MCMCandMixing}

\subsection{MCMC Methods}
\label{subsec: MCMC_review}
Let \(P\) denote a local Markov kernel that leaves the target posterior \(\pi\) invariant. We consider Metropolis--Hastings, Gibbs, fixed-length HMC \citep{neal2011mcmc}, and NUTS \citep{hoffman2014no}, with tuning held fixed during sampling. These methods can explore one posterior region effectively while moving slowly between separated regions or along extended flat directions. Algorithms~\ref{alg: HMC} and~\ref{alg: MHMC} in the Online Appendix record the baseline conventions. SMC, described in Online Appendix~\ref{appendix: SMC}, provides a population-based benchmark in the simulations.

\subsection{Identification-aware MCMC}
\label{subsec: identification_driven}
Let \(\Omega\) denote the parameter space. Given the observed sample
\(Y_{1:N}=(Y_1,\ldots,Y_N)\), the target posterior is
\[
\pi_N(\theta)\ \propto\
p_{0}(\theta)\,p_N(Y_{1:N}\mid\theta),
\]
where \(p_{0}\) is the prior density and \(p_N\) is the likelihood. We write \(\pi\) for \(\pi_N\), since \(N\) remains fixed throughout the posterior analysis.

To define observational equivalence independently of the observed sample size, let \(P_\theta^{(n)}\) denote the joint law of the first \(n\) observables under \(\theta\). Parameters \(\theta\) and \(\theta'\) are observationally equivalent if
\[
P_\theta^{(n)}=P_{\theta'}^{(n)}
\qquad\text{for every }n\ge1.
\]
We write \(K(\theta)\subseteq\Omega\) for the observational equivalence class of \(\theta\).

If these laws have densities with respect to a common dominating measure \(m_n\), observational equivalence implies equality of the densities for \(m_n\)-almost every sample. In the models below, the likelihoods agree for every possible sample, so \(p_N(Y_{1:N}\mid\theta)\) is constant on \(K(\theta)\).

To construct the teleport kernel, assume that \(\Omega\) is a standard Borel space and choose in advance a measurable map \(\chi\) from \(\Omega\) into another standard Borel space. Define
\[
\mathcal C(\theta)
:=
\{\theta'\in\Omega:\chi(\theta')=\chi(\theta)\},
\qquad
\mathcal C(\theta)\subseteq K(\theta).
\]
Thus the level sets of \(\chi\) form a measurable refinement of the observational equivalence partition, and \(\mathcal C(\theta)\) is the level set containing \(\theta\).

A regular conditional distribution of \(\pi\) given \(\chi\) exists. A teleport kernel can be constructed by sampling directly from this conditional law:
\begin{equation}
\label{eqn:teleport_kernel}
T(\theta,A)
=
\pi\!\left(A\mid\chi=\chi(\theta)\right).
\end{equation}
For \(\pi\)-almost every \(\theta\), this conditional law is concentrated on \(\mathcal C(\theta)\). Choosing \(\mathcal C(\theta)=K(\theta)\)
allows teleportation over the full observational equivalence class. A finer partition restricts teleportation to a subset of that class while preserving the target posterior.

When direct sampling is unavailable, the teleport step can instead use a Markov update on \(\mathcal C(\theta)\) that is reversible with respect to \(\pi(\,\cdot\mid\chi=\chi(\theta))\). Such an update preserves the conditional law but need not produce an
independent draw from it. Appendix~\ref{append:mtm_teleport} gives a multiple-try
Metropolis (MTM) construction, and Online Appendix~\ref{append:empirical} gives the conditional pool updates used in the applications.\footnote{The set of states considered by
a teleport update may also be chosen at random. If the joint set-and-candidate proposal is well defined, has reverse support, and has evaluable forward and reverse probabilities, an MH or MTM correction using the full joint probabilities can produce a \(\pi\)-reversible update. The spectral results below then apply when their remaining assumptions hold.}

Although we focus on observational equivalence, the construction by conditional sampling and the general invariance, reversibility, and spectral-gap results apply to any partition induced by a measurable map, provided their stated assumptions hold. Observational equivalence is both theoretically important and common in applications: it is determined by the model, applies across sample sizes, does not depend on the realized data, and often has a clear structural interpretation. More general partitions can also be used, including those based on the \(N\)-sample law or the likelihood for the realized data. For instance, in structural vector autoregressions with narrative
restrictions, the unconditional likelihood has flat regions bounded
by data-dependent trimming points \citep{giacomini2023_NR}.

To motivate and illustrate our proposal, consider the following toy example with 4 states.
\begin{example}
\label{ex:toy}
Consider two parameters \(\theta_1\) and \(\theta_2\), each taking binary values 0 and 1. Write \(p_{ij}\) for the posterior probability of state \((i,j)\).
Assume the observable law depends on \((\theta_1,\theta_2)\) only through \(\theta_1-\theta_2\). Then \((0,0)\) and \((1,1)\) are observationally equivalent. Under a uniform conditional prior on that class, this implies \(p_{00}=p_{11}\).

We construct a Gibbs transition matrix \(P\) over the four states
\[
(0,0),\ (0,1),\ (1,0),\ (1,1),
\]
listed in that order. Starting from \((0,0)\), we run 10,000 Gibbs iterations under
\[
p_{11}=p_{00}=0.49999,\qquad p_{10}=p_{01}=0.00001,
\]
to simulate an extremely bimodal scenario.\footnote{The equality \(p_{11}=p_{00}\) follows from observational equivalence and the uniform conditional prior on the diagonal class. The equality \(p_{10}=p_{01}\) is imposed only for computational simplicity and is not implied by observational equivalence.} As shown in Figure~\ref{fig:toy_ex1}, the displayed local Gibbs run remains at \((0,0)\) and does not visit \((1,1)\).\footnote{Extending the same seeded run to one million iterations gives empirical masses \(0.522\) at \((0,0)\) and \(0.478\) at \((1,1)\). Thus, even after one million iterations, the empirical state probabilities differ from their target values by as much as \(0.022\) in this simple setup.}
\end{example}
\begin{figure}[!htbp]
  \centering
  \begin{subfigure}[t]{0.48\textwidth}
    \centering
    \includegraphics[width=\linewidth]{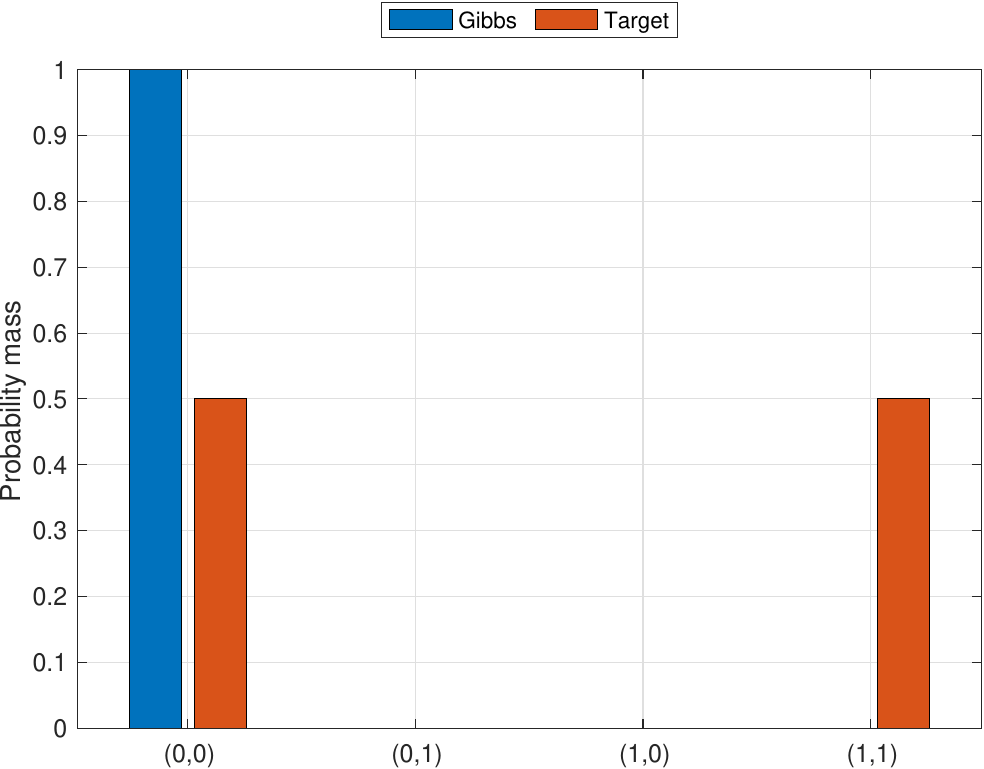}
    \caption{Gibbs-sampled distribution.}
    \label{fig:toy_ex1}
  \end{subfigure}\hfill
\begin{subfigure}[t]{0.48\textwidth}
  \centering
  \includegraphics[width=\linewidth]{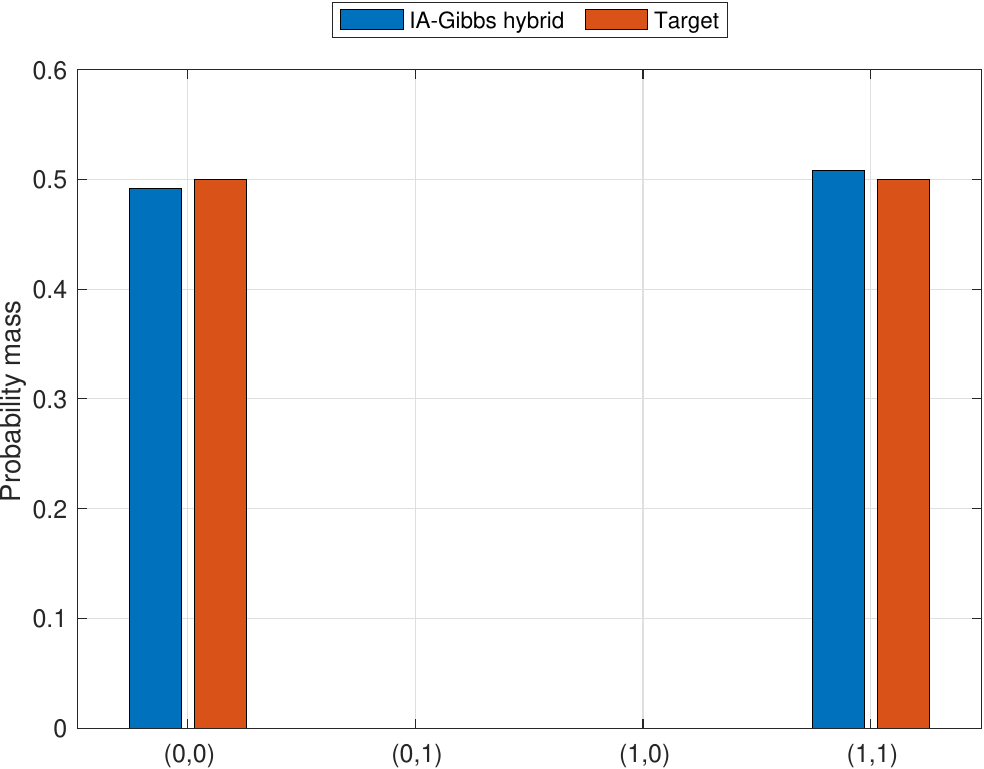}
  \caption{Hybrid identification-aware Gibbs.}
  \label{fig:toy_ex1_modified}
\end{subfigure}
\caption{Gibbs and hybrid identification-aware Gibbs over 10,000 iterations. The hybrid uses \(\varepsilon=0.5\).}
  \label{fig:toy_ex1_compare}
\end{figure}

A natural remedy is to make the Gibbs sampler \textit{identification–aware} by adding occasional teleport moves between observationally equivalent states to the local Gibbs sampler.

Let \(P\) be a \(\pi\)-invariant local kernel, such as MH, HMC, Gibbs, or a NUTS transition whose tuning is held constant during sampling, and let \(T\) be a \(\pi\)-preserving teleport kernel. The theory below keeps \(P\), \(T\), and \(\varepsilon\) unchanged across iterations. We study
\begin{equation}
\label{eq:hybrid_kernel}
\widetilde P_H=(1-\varepsilon)P+\varepsilon T,
\qquad \varepsilon\in(0,1).
\end{equation}
At each iteration, the chain applies either the local kernel or the teleport kernel.

\begin{algorithm}[Identification-aware hybrid kernel]
\phantomsection
\label{alg:hybrid_ia}
\begin{enumerate}
\item Given a target distribution \(\pi\), Markov kernels \(P\) and \(T\) that preserve \(\pi\), an initial state \(\theta_0\), a teleport probability \(\varepsilon\in(0,1)\), and a number of iterations \(M\).
\item For \(t=1\) to \(M\):
\begin{enumerate}
\item draw \(B_t\sim\mathrm{Bernoulli}(\varepsilon)\);
\item if \(B_t=1\), draw \(\theta_t\sim T(\theta_{t-1},\cdot)\);
\item if \(B_t=0\), draw \(\theta_t\sim P(\theta_{t-1},\cdot)\).
\end{enumerate}
\item Return \(\{\theta_t\}_{t=1}^M\).
\end{enumerate}
\end{algorithm}

\addtocounter{example}{-1}
\begingroup
\renewcommand{\theHexample}{\arabic{example}.continued}
\begin{example}[Continued]
In the state order given above, define
\[
T=
\begin{pmatrix}
1/2 & 0 & 0 & 1/2\\
0 & 1 & 0 & 0\\
0 & 0 & 1 & 0\\
1/2 & 0 & 0 & 1/2
\end{pmatrix}.
\]
Thus \(T\) draws uniformly from the observational equivalence class \(\{(0,0),(1,1)\}\) and leaves the off-diagonal states unchanged. Combining this teleport kernel with the original Gibbs kernel through \eqref{eq:hybrid_kernel} gives the identification-aware chain in Figure~\ref{fig:toy_ex1_compare}. Online Appendix~\ref{oa:gibbs_details} compares the finite-state absolute gaps. As the off-diagonal posterior mass vanishes, the local Gibbs gap tends to zero while the hybrid gap remains positive.
\end{example}
\endgroup
When \(\theta=(u,v)\) and varying \(v\) with \(u\) fixed leaves the observable law unchanged, teleportation can be implemented as a Gibbs update:
\[
T\big((u,v),A\big)
=
\pi(A\mid u).
\]
When the local update is Metropolis, the resulting construction can be viewed as Metropolis-within-Gibbs. The same logic extends to equivalence sets that are curved, discrete, or only implicitly characterized.

\begin{example}[Teleport moves in sign-restricted SVARs]
This interpretation also clarifies an existing practice in sign-restricted SVAR computation. Conditional on the reduced-form parameters, the likelihood is unchanged over the admissible orthogonal rotations. Uhlig-style samplers \citep{uhlig2005effects,rubio2010structural} draw the rotation conditional on the reduced form. In our notation, these rotation draws are teleport moves between observationally equivalent structural representations. Related robust Bayesian procedures, such as \citet{giacomini2021robust}, likewise condition on the reduced form and integrate over, or place priors on, these representations.

Although these rotations leave the likelihood unchanged, their relative weights depend on the prior.  \citet{baumeister2015sign} emphasize that Haar-uniform rotation draws induce substantive priors over structural impulse responses. Teleport moves must therefore preserve the conditional distribution implied by the specified prior.
\end{example}

Appendix~\ref{append: ia_variants} presents the random-order kernel \(\bar P=\tfrac12(TP+PT)\) and batch augmentation, which uses additional draws from \(T(\theta_t,\cdot)\) to estimate posterior expectations. The main spectral-gap arguments use the hybrid kernel \(\widetilde P_H\).

\begin{proposition}[Invariance of identification-aware variants]
\label{prop:stationary}
Let \(P\) be a Markov kernel with invariant distribution \(\pi\), and let \(T\) be a \(\pi\)-preserving teleport kernel. Then the sequential kernels \(PT\) and \(TP\), as well as the hybrid kernel
\[
\widetilde P_H=(1-\varepsilon)P+\varepsilon T,
\]
are \(\pi\)-invariant. If, in addition, both \(P\) and \(T\) are \(\pi\)-reversible, then \(\widetilde P_H\) is \(\pi\)-reversible.
\end{proposition}

Proposition~\ref{prop:stationary} shows that teleport moves can be combined with the local sampler without changing the target posterior. For the kernel defined by conditional sampling in \eqref{eqn:teleport_kernel}, Lemma~\ref{lemma:conditional_sampling} establishes that \(T\) is \(\pi\)-reversible.

\subsection{Convergence and Mixing Rates}
\label{subsec: convergence}
Write \(\pi(f)=\int f\,d\pi\), let \(L_p(\pi)\) denote the real \(p\)-integrable functions for \(p=1,2\), and set \(\|f\|_2^2=\pi(f^2)\). Under standard general-state-space ergodicity conditions \citep{roberts2004general,meyn2009markov}, the chain converges to \(\pi\), and posterior averages \(\widehat\pi_M(f)=M^{-1}\sum_{t=1}^M f(\theta_t)\)  converge almost surely to \(\pi(f)\) for \(f\in L_1(\pi)\). For reversible \(P\), define \(L_2^0(\pi)=\{f\in L_2(\pi):\pi(f)=0\}\). Our main rate criterion is the Poincar\'e spectral gap
\[
\gamma(P)
=1-\sup_{\substack{f\in L_2^0(\pi)\\\|f\|_2=1}}
\langle Pf,f\rangle_\pi,
\]
where \(Pf(\theta)=\int f(\theta')P(\theta,d\theta')\) and \(\langle f,g\rangle_\pi=\int fg\,d\pi\). For an irreducible finite-state reversible chain, \(\gamma(P)=1-\lambda_2\), where \(\lambda_2\) is its second-largest eigenvalue. If \(P\) has nonnegative spectrum, its exact worst-case \(L_2(\pi)\) contraction factor after \(n\) steps is \((1-\gamma(P))^n\).

When \(\gamma(P)>0\), the central limit theorem for stationary ergodic reversible chains applies to \(f\in L_2(\pi)\) \citep[Corollary~1.5]{kipnis1986central}. Writing \(\sigma_{P,f}^2\) for its asymptotic variance, the spectral theorem gives
\[
\mathcal D(P)
:=\sup_{\substack{f\in L_2(\pi)\\\operatorname{Var}_\pi(f)>0}}
\frac{\sigma_{P,f}^2}{\operatorname{Var}_\pi(f)}
=\frac{2}{\gamma(P)}-1.
\]
This is the worst-case ratio of chain draws to independent draws needed for the same asymptotic precision. A uniform positive gap therefore controls posterior-average precision and, for nonnegative
spectrum, the worst-case \(L_2(\pi)\) contraction factor. Throughout, \(\gamma(P)\) denotes the Poincar\'e gap. The finite-state absolute gap in Online
Appendix~\ref{oa:gibbs_details} is written as
\(\gamma_{\mathrm{abs}}(P)\).

Proposition~\ref{prop:deltaMH} studies increasingly separated posterior regions, while Proposition~\ref{prop:manifold} studies targets whose observational equivalence fibers grow in diameter. Their proofs are in Appendix~\ref{append:proof_main}. Online Appendix~\ref{append:additional_results} gives further calculations and extensions to other kernels.

\subsubsection{Multimodality}
 
Multimodality is well known in mixture models, where different allocations or labels can produce distinct posterior modes \citep{fruhwirth2006finite}. In macroeconomics, it can arise in high-dimensional models \citep{waggoner2016striated} and from normalization in structural VARs \citep{hamilton2007normalization}. SVAR restrictions may also identify parameters locally while leaving several globally admissible solutions \citep{bacchiocchi2025locally}. These examples motivate our focus on movement between separated posterior regions.

We begin with a simple example inspired by \citet{guan2007small}. A circle avoids the endpoint treatment required on a bounded interval and isolates the difficulty of moving between two separated modes.
\begin{example}
\label{example:1dmultimodal}
The state space $\Omega$ is a one-dimensional circle of circumference $4L$, represented as the interval $[-2L, 2L]$ with its endpoints joined. The target distribution $\pi(\theta)$ is bimodal:
\[
\pi(\theta) \propto
\begin{cases}
 \exp(-\nu |\theta|), & \text{if } \theta \in [-L, L] \\
 \exp(-\nu (2L - |\theta|)), & \text{if } \theta \in [-2L, -L) \cup (L, 2L]
\end{cases}
\]

where $L \gg 1$ and $\nu > 0$. The parameters $L$ and $\nu$ control the difficulty of the sampling problem: large $L$ increases the separation between modes, while large $\nu$ makes each mode sharper and more concentrated.

\end{example}

A local random-walk Metropolis (RWM) kernel with a uniform \(\delta\)-ball proposal mixes poorly on this two-mode circle. As shown in Proposition~\ref{prop:1ddeltaball} in the Online Appendix, its spectral gap has the upper bound
\[
\gamma(P)\ \le\ C_1 e^{-\nu (L-\delta)}.
\]
The upper bound decreases as the modes become sharper or farther apart, explaining why a local chain can remain near one mode for many iterations.

To isolate how observational equivalence information changes this behavior, suppose antipodal points are observationally equivalent. Let
\[
s(\theta) =
\begin{cases}
\theta+2L, & \theta\in[-2L,0),\\[2pt]
\theta-2L, & \theta\in[0,2L)
\end{cases}
\]
be the antipodal shift. The observational equivalence class is
\[
K(\theta)=\{\theta,s(\theta)\}.
\]
By symmetry, the posterior conditional law on \(K(\theta)\) is uniform, so the teleport kernel is
\[
T(\theta,\cdot)
=
\frac12\delta_\theta(\cdot)
+
\frac12\delta_{s(\theta)}(\cdot),
\]
where \(\delta_\theta\) denotes the Dirac measure at \(\theta\). Let \(P\) be one RWM step with the \(\delta\)-ball proposal, and define the hybrid identification-aware transition
\[
\widetilde P_H=(1-\varepsilon)P+\varepsilon T,
\qquad \varepsilon\in(0,1).
\]
The number \(1/2\) in \(T\) comes from the conditional posterior on the two-point equivalence class, while \(\varepsilon\) controls how often the chain attempts a teleport. Thus the one-step probability of moving from one mode to its antipodal counterpart is \(\varepsilon/2\). The same proposition also shows that, for \(\delta\) of the same order as \(\nu^{-1}\), the resulting spectral-gap lower bound is independent of the mode separation \(L\). The example illustrates the same separation mechanism studied more generally in Proposition~\ref{prop:deltaMH}.

To extend the example, let \(\pi\) be a probability density on \(\Omega\subseteq\mathbb R^n\), and partition \(\Omega\) into measurable regions \(A_1,\ldots,A_m\), with \(m\ge2\) and \(\pi(A_i)>0\). For each \(i\), define
\[
w_i:=\pi(A_i),
\qquad
\pi_i(B):=\frac{\pi(B\cap A_i)}{w_i},
\qquad
A_i^c:=\Omega\setminus A_i.
\]
Let \(\mu_i\in A_i\) be a mode of \(\pi\) restricted to \(A_i\), and assume
\[
d_i:=\operatorname{dist}(\mu_i,A_i^c)>0,
\qquad
d_*:=\min_i d_i.
\]
Because \(\mu_j\in A_i^c\) whenever \(i\ne j\), growth of \(d_*\) implies growth of every pairwise distance \(\|\mu_i-\mu_j\|\). The parameter \(d_*\) therefore captures the increasingly separated regime studied below.

\begin{assumption}
\phantomsection
\label{assump:multimodality}
\begin{enumerate}[label=(\roman*)]
\item \label{assump:multimodality1}
There exist constants $c_i\ge1$ and
$\nu_i>0$ such that, for all $r\in(0,d_i]$,
\[
\pi_i \big(\{\theta\in A_i:\ \|\theta-\mu_i\|\ge r\}\big)
  \le  c_i e^{-\nu_i r}.
\]
Set $c_{\max}:=\max_i c_i$ and $\nu_{\min}:=\min_i \nu_i$.
\item \label{assump:multimodality2}
Fix a step size $\delta>0$ and define
$A_i^{\mathrm{int}}:=\{\theta\in A_i:\ \operatorname{dist}(\theta,A_i^c)\ge\delta\}$.
There exist $\varepsilon_1\in(0,1]$, $\varepsilon_2>0$, and measurable cores
$A_i^\circ\subset A_i$ with $\pi_i(A_i^\circ)\ge\varepsilon_1$ such that, for
every $i\neq j$ and all $\theta\in A_i^\circ$,

$$
T(\theta,A_j^{\mathrm{int}})\ \ge\ \varepsilon_2,
$$
where \(T\) is the teleport kernel used in the hybrid mixture \(\widetilde P_H\).

\item\label{assump:multimodality3}
For each \(i\), let \(P_{A_i}\) denote the restriction of the local kernel \(P\) to \(A_i\). There exist \(n_0\in\mathbb N\), \(\eta_0>0\), and a probability measure \(m_i\) supported on a measurable set \(C_i\subset A_i\) with \(\pi_i(C_i)>0\) such that
\[
\inf_{\theta\in A_i} P_{A_i}^{\,n_0}(\theta,\cdot)\ge \eta_0m_i(\cdot).
\]
\end{enumerate}
\end{assumption}

Assumption~\ref{assump:multimodality}\ref{assump:multimodality1} is a mode-wise concentration condition: within each region \(A_i\), it bounds the tail mass away from the mode \(\mu_i\) up to the radius \(d_i\). It is used to upper-bound the conductance of the local RWM kernel by showing that points in \(A_i\) within distance \(\delta\) of \(A_i^c\) have exponentially small \(\pi_i\)-mass as \(d_i\) grows. Assumption~\ref{assump:multimodality}\ref{assump:multimodality2} states that the teleport kernel places a fixed fraction of probability into the \(\delta\)-interior \(A_j^{\mathrm{int}}\) of any other region, uniformly over a positive-mass core \(A_i^\circ\). This supplies the inter-mode communication for the hybrid kernel \(\widetilde P_H=(1-\varepsilon)P+\varepsilon T\). Assumption~\ref{assump:multimodality}\ref{assump:multimodality3} provides a within-mode Doeblin-type minorization for a finite number of local moves. This yields a uniform positive lower bound on the restricted local-chain gap inside each \(A_i\) \citep{meyn2009markov}.

Log-concavity of the within-region law \(\pi_i\) is one way to verify this exponential-tail condition. \citet[Theorem~3.1]{borell1974convex} gives exponential integrability of the Euclidean norm and hence an exponential tail around \(\mu_i\). The resulting rate, multiplicative constant, and any uniformity over a sequence must be checked separately. Verifying the within-mode minorization in Assumption~\ref{assump:multimodality}\ref{assump:multimodality3}, however, also requires local conductance or support-regularity conditions that rule out unrelated bottlenecks inside a single modal region.

A classical example is a finite Gaussian mixture with label switching. Suppose the posterior has \(m!\) well-separated modes corresponding to permutations of the component labels. The regions \(A_i\) can be chosen around these modes. Under standard local regularity, each conditional target \(\pi_i\) has exponentially decaying tails, as required by Assumption~\ref{assump:multimodality}\ref{assump:multimodality1}. Label permutations map a point near one labeling mode to the corresponding point near another. If the conditional posterior probabilities of the permutations are bounded away from zero on the stated cores, Assumption~\ref{assump:multimodality}\ref{assump:multimodality2} follows.

For Assumption~\ref{assump:multimodality}\ref{assump:multimodality3}, bounded proposal and target densities on a compact core give a small-set minorization on that core. The assumption as stated also requires every starting point in \(A_i\) to reach such a core in a uniformly bounded number of restricted steps. That additional reachability condition must be verified for the chosen regions. Subject to it, label switching provides the direct communication between posterior modes used in Proposition~\ref{prop:deltaMH}.

\begin{proposition}[Separated modal regions]
\label{prop:deltaMH}
Let \(\pi\) be a probability density on \(\Omega\subseteq\mathbb R^n\), and suppose Assumption~\ref{assump:multimodality} holds with sets \(A_1,\ldots,A_m\), modes \(\mu_i\), radii \(d_i\), and a \(\pi\)-reversible teleport kernel \(T\). Let \(P\) be the RWM kernel with the uniform \(\delta\)-ball proposal,
\[
q(\theta,\cdot)=\mathrm{Unif}\bigl(B(\theta,\delta)\bigr),
\]
and let
\[
\widetilde P_H=(1-\varepsilon)P+\varepsilon T,
\qquad \varepsilon\in(0,1).
\]
Then:
\begin{enumerate}
\item For any \(\delta>0\),
\[
\gamma(P)\le 2c_{\max}\exp\{-\nu_{\min}(d_*-\delta)_+\},
\]
where \((x)_+:=\max\{x,0\}\).

\item There exists \(c_0>0\), depending only on \((\varepsilon,\varepsilon_1,\varepsilon_2,n_0,\eta_0,m)\) and not on the separations \(d_i\), such that
\[
\gamma(\widetilde P_H)\ge c_0.
\]
\end{enumerate}
\end{proposition}
See Appendix~\ref{append:proof_deltaMH} for a detailed proof.

Across targets for which the constants in Assumption~\ref{assump:multimodality} can be chosen uniformly, Proposition~\ref{prop:deltaMH} shows that the \(\delta\)-ball RWM gap can deteriorate exponentially as \(d_*\) grows. Since growth of \(d_*\) implies growing pairwise separation of the mode centers, this captures the inter-mode separation bottleneck. The hybrid identification-aware RWM (IA--RWM) kernel \(\widetilde P_H\) retains a positive lower bound that does not depend on \(d_*\).

Writing \(b=\nu_{\min}(d_*-\delta)\) for \(d_*>\delta\),
the identity \(\mathcal D(P)=2/\gamma(P)-1\) gives, when
\(\gamma(P)>0\),
\[
\mathcal D(P)\ge \frac{e^b}{c_{\max}}-1,
\qquad
\mathcal D(\widetilde P_H)\le \frac{2}{c_0}-1.
\]
Thus the worst-case variance multiplier can grow exponentially for local RWM while remaining uniformly bounded for IA--RWM. If \(\gamma(P)=0\), no finite worst-case variance bound follows.

Teleport moves may use only a subset of the observationally equivalent representations. The spectral comparison requires sufficient communication between modal regions, together with mixing by the local kernel within each region. Corollary~\ref{cor:k_mode_random_subset} in the Online Appendix gives a reversible equal-mass construction using a randomly selected set of \(k\) regions and establishes an explicit gap for movement between regions.

The Online Appendix also gives extensions for other local kernels. Proposition~\ref{prop:gaussian_rw} gives the analogous separation comparison for Gaussian RWM under explicit conditions on communication through the teleport kernel and on within-region gaps. Proposition~\ref{prop:HMC} gives the fixed-length HMC result under a displacement-tail condition whose constants do not change with inter-mode separation. Corollary~\ref{cor:fixed_nuts} gives the corresponding conclusion for a fixed-tuning \(\pi\)-reversible NUTS transition whose selected next state satisfies an analogous bound. Under the conditions stated in those results, the corresponding identification-aware hybrid retains a positive gap lower bound independent of the inter-mode separation.

\subsubsection{Local non-identification}
\label{subsubsec: setID}

Apart from multimodality, set identification can also create long observational equivalence fibers. The likelihood is constant along each fiber, and a local random walk traverses a long fiber only diffusively. We first show this bottleneck for a product target on a rectangle and then turn to curved fibers.

\begin{example}
\label{example:flatY}
Let
\[
\Theta_D=X\times Y_D\subset\mathbb R^d,
\qquad
Y_D=[-D,D]^r,
\qquad
d=m+r,
\]
where \(X\subset\mathbb R^m\) is fixed and bounded. Let \(u_D\) be the uniform density on \(Y_D\), and consider
\[
\pi_D(\theta_x,\theta_y)=p(\theta_x)u_D(\theta_y),
\]
where \(p\) is a continuous probability density on \(X\), bounded away from zero and infinity. The observable law depends on \(\theta_x\) but not on \(\theta_y\), so each set \(\{\theta_x\}\times Y_D\) is an observational equivalence fiber.
\end{example}

As \(D\) grows, the marginal distribution of \(\theta_x\) remains fixed, while each fiber grows in diameter. Under the conditions of Online Appendix Proposition~\ref{prop:local_unid_flatY_bar}, the local RWM gap tends to zero, whereas the hybrid IA--RWM kernel in \eqref{eq:hybrid_kernel} retains a positive lower bound.

This idea extends beyond the product example above. The next result covers curved and irregular fibers under explicit geometric conditions. A balanced cut must place little probability near its boundary relative to the fiber diameter. The teleport and local kernels must also reach a common neighborhood in a uniformly bounded number of steps.

\begin{assumption}
\phantomsection
\label{assump:local_nonid}
\begin{enumerate}[label=(\roman*)]

\item \label{assump:local_nonid1}
Let \(\Theta\subseteq\mathbb R^d\) be a Borel set and fix an integer \(1\le r<d\). There exists a continuous surjection
\[
\phi:\Theta\to\Phi\subset\mathbb R^{d-r},
\]
where \(\Phi\) is compact and path-connected, such that each fiber
\[
\mathcal F(u):=\{\theta\in\Theta:\phi(\theta)=u\}
\]
is a compact, connected, \(r\)-dimensional \(C^1\) embedded submanifold of \(\mathbb R^d\) contained in \(\Theta\), possibly with boundary. Each fiber is contained in an observational equivalence class:
\[
\mathcal F(\phi(\theta))\subseteq K(\theta)
\qquad\text{for all }\theta\in\Theta.
\]
Define
\[
D(u):=\sup_{\theta,\theta'\in\mathcal F(u)}d_{\mathcal F}(\theta,\theta'),
\]
where \(d_{\mathcal F}\) is intrinsic distance along the fiber.

\item \label{assump:local_nonid2}
The posterior admits the disintegration
\[
\pi(d\theta)=m(du)\,\pi(d\theta\mid\phi=u),
\qquad u=\phi(\theta),
\]
where \(m\) has a density \(\bar m\) with respect to \((d-r)\)-dimensional Lebesgue measure restricted to \(\Phi\), with
\[
0<m_{\min}\le \bar m(u)\le m_{\max}<\infty.
\]
The conditional law \(\pi(\,\cdot\mid\phi=u)\) is absolutely continuous with respect to the normalized \(r\)-dimensional Hausdorff measure \(\mu_u\) on \(\mathcal F(u)\):
\[
\pi(d\theta\mid\phi=u)=w_u(\theta)\mu_u(d\theta),
\qquad
0<w_{\min}\le w_u(\theta)\le w_{\max}<\infty
\quad \mu_u\text{-a.e.}
\]
In this setting, the teleport kernel draws directly from the conditional law on the current fiber:
\[
T(\theta,A):=\pi(A\mid\phi=\phi(\theta)).
\]

\item \label{assump:local_nonid3}
There exists a measurable set \(S\subseteq\Theta\) and constants \(0<c_-\le c_+<1\), \(C_S<\infty\), and \(a_0>0\) such that, writing \(S_u:=S\cap\mathcal F(u)\) and \(\partial_{\mathcal F}S_u\) for the boundary of \(S_u\) within \(\mathcal F(u)\),
\[
c_-\le \pi(S_u\mid\phi=u)\le c_+,
\]
and
\[
\pi\!\left(
\left.
\left\{\theta\in\mathcal F(u):
d_{\mathcal F}(\theta,\partial_{\mathcal F}S_u)\le s
\right\}
\,\right|\,\phi=u
\right)
\le C_S\min\left\{1,\frac{s}{D(u)}\right\}
\qquad \forall s\ge0,\ \forall u\in\Phi.
\]
Moreover, for every \(s\ge0\) and every \(\theta,\theta'\in\Theta\), if
\[
d_{\mathcal F}\bigl(\theta,\partial_{\mathcal F}S_{\phi(\theta)}\bigr)>s
\quad\text{and}\quad
\|\theta'-\theta\|\le a_0s,
\]
then
\[
\mathbf 1_S(\theta')=\mathbf 1_S(\theta).
\]

\item \label{assump:local_nonid4}
There exists a measurable map \(s:\Phi\to\Theta\) satisfying \(s(u)\in\mathcal F(u)\) for every \(u\in\Phi\), constants \(\rho>0\), \(\kappa_T>0\), a point \(u^\star\in\Phi\), and a compact set \(\mathcal K_0\subset\Theta\) such that the \(\rho/4\)-neighborhood of \(\mathcal K_0\) is contained in \(\Theta\),
\[
\bigcup_{u\in\Phi}B(s(u),\rho)\subseteq \mathcal K_0,
\]
and
\[
\pi\bigl(B(s(u),\rho/16)\cap\mathcal F(u)\mid\phi=u\bigr)\ge \kappa_T
\qquad \forall u\in\Phi.
\]
The target has a Lebesgue density \(p\) on the \(\rho/4\)-neighborhood of \(\mathcal K_0\), bounded there by
\[
0<p_{\min}\le p(\theta)\le p_{\max}<\infty.
\]
In addition, every \(s(u)\) can be connected to \(s(u^\star)\) by a polygonal path inside \(\mathcal K_0\) with at most \(N_0\) segments, each of Euclidean length at most \(\rho/16\).

\item \label{assump:local_nonid5}
The local proposal density \(g\) is symmetric and translation-invariant, and there exists \(c_g>0\) such that
\[
g(z)\ge c_g
\qquad \text{for all }\|z\|\le \rho/4.
\]
Its tail probability
\[
\overline G(t):=\int_{\|z\|>t}g(z)\,dz
\]
has finite first moment:
\[
\int_0^\infty \overline G(t)\,dt<\infty.
\]
\end{enumerate}
\end{assumption}

The fibers of \(\phi\) form a measurable refinement of the observational equivalence partition. The map need not be an economically meaningful reduced-form parameter. The inclusion in Assumption~\ref{assump:local_nonid}\ref{assump:local_nonid1} makes the likelihood constant along each fiber.

Assumption~\ref{assump:local_nonid}\ref{assump:local_nonid2} disintegrates the posterior into a marginal law across fibers and the conditional law \(\pi(\,\cdot\mid\phi=u)\) on each fiber. It specifies the teleport kernel \(T(\theta,A)=\pi(A\mid\phi=\phi(\theta))\). The bounds on \(w_u\) keep this conditional law uniformly comparable to normalized Hausdorff measure.

Assumption~\ref{assump:local_nonid}\ref{assump:local_nonid3} gives the bottleneck condition for local RWM. The probability within distance \(s\) of a balanced cut is bounded above by \(C_S\min\{1,s/D(u)\}\), and the final condition in item~\ref{assump:local_nonid3} prevents a short Euclidean proposal from crossing the cut away from its boundary.

Assumption~\ref{assump:local_nonid}\ref{assump:local_nonid4} ensures that, with probability at least \(\kappa_T\), the teleport lands near \(s(u)\), after which the local kernel can follow a uniformly short path to the \(\rho/16\)-neighborhood of \(s(u^\star)\). Assumption~\ref{assump:local_nonid}\ref{assump:local_nonid5} bounds the proposal density near the origin and controls its tails.

\begin{proposition}[Curved observational equivalence fibers]
\label{prop:manifold}
Let Assumptions~\ref{assump:local_nonid}\ref{assump:local_nonid1}--\ref{assump:local_nonid5} hold. Let \(P\) be the random-walk Metropolis kernel on \(\Theta\) with proposal density \(g\), let \(T\) be the teleport kernel specified in Assumption~\ref{assump:local_nonid}\ref{assump:local_nonid2}, and define
\[
\widetilde P_H=(1-\varepsilon)P+\varepsilon T,
\qquad \varepsilon\in(0,1).
\]
Set
\[
\mathfrak D^{-1}:=\int_{\Phi}\frac{1}{D(u)}\,m(du).
\]
Then:
\begin{enumerate}
\item There exists \(C<\infty\), depending only on the constants in Assumptions~\ref{assump:local_nonid}\ref{assump:local_nonid1}--\ref{assump:local_nonid5}, such that
\[
\gamma(P)\le C\,\mathfrak D^{-1}.
\]
In particular, if \(D(u)\ge D\) on the support of \(m\), then \(\gamma(P)\le C/D\).

\item There exists \(c_0>0\), depending only on \(\varepsilon\) and the constants in Assumptions~\ref{assump:local_nonid}\ref{assump:local_nonid2}--\ref{assump:local_nonid5}, such that
\[
\gamma(\widetilde P_H)\ge c_0.
\]
\end{enumerate}
\end{proposition}

The proof is in Appendix~\ref{append:proof_manifold}. Part~(1) shows the upper bound on the local RWM spectral gap decreases as the equivalence fibers grow. Part~(2) establishes a common positive lower bound for the hybrid when the constants in Assumption~\ref{assump:local_nonid}\ref{assump:local_nonid2}--\ref{assump:local_nonid5}, including the lower bound \(\kappa_T\), remain uniform. The expanding uniform cylinder does not satisfy this fixed-radius condition: the conditional mass of such a neighborhood tends to zero as the cylinder lengthens. Proposition~\ref{prop:local_unid_flatY_bar} in the Online Appendix treats that product case separately by sampling directly from the full within-fiber conditional law.

\section{Simulations}
\label{sec:simulation}
The simulations illustrate the proposed approach in the two settings studied above: separated observationally equivalent modes and likelihood-flat directions.
\begin{simulation}[Gaussian mixture]
We draw a sample of size \(1{,}000\) from
\[
X\sim p\mathcal N(\mu_1,\sigma_1^2)+(1-p)\mathcal N(\mu_2,\sigma_2^2),
\]
with \((\mu_1,\mu_2,\sigma_1,\sigma_2,p)=(0,20,1,5,0.3)\). We use independent uniform priors on \(\mu_i\in[-50,50]\), \(\sigma_i\in[0.1,20]\), \(i=1,2\), and \(p\in[0,1]\). We compare Gaussian random-walk Metropolis (RWM) with identification-aware RWM (IA--RWM) using the label permutation
\[
(\mu_1,\mu_2,\sigma_1,\sigma_2,p)
\longleftrightarrow
(\mu_2,\mu_1,\sigma_2,\sigma_1,1-p).
\]
\end{simulation}
Because the permutation leaves the mixture likelihood unchanged, the two labelings are observationally equivalent. Figure~\ref{fig:density_chain_old} and \ref{fig:density_chain_new} show one representative comparison. The local RWM chain remains on one labeling, while the IA--RWM chain visits both. Figure~\ref{fig:MCMC_moment} and~\ref{fig:IAMH_moment} plot the within-chain averages of $(\mu_1,\mu_2)$ and $(\sigma_1,\sigma_2)$ across 1,000 runs of each sampler on the same dataset, with each point representing one run. Before symmetrization, these moments depend on the labels, so their variation shows whether a chain visits both representations. Label-invariant summaries are unchanged by relabeling.
\begin{figure}[!htbp]
  \centering

  \begin{subfigure}[b]{0.49\textwidth}
    \centering
    \includegraphics[width=\textwidth]{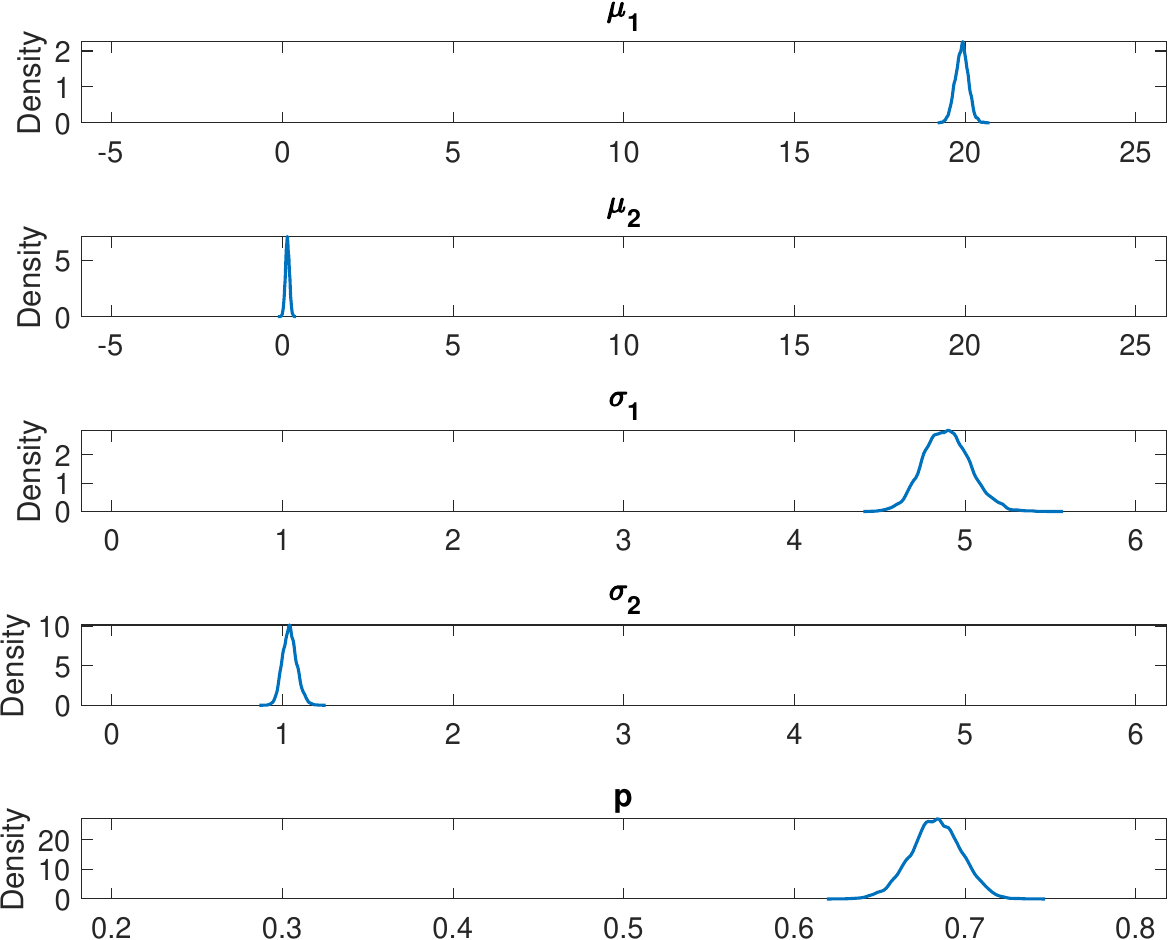}
    \caption{Gaussian RWM: representative chain}
    \label{fig:density_chain_old}
  \end{subfigure}
  \hfill
  \begin{subfigure}[b]{0.49\textwidth}
    \centering
    \includegraphics[width=\textwidth]{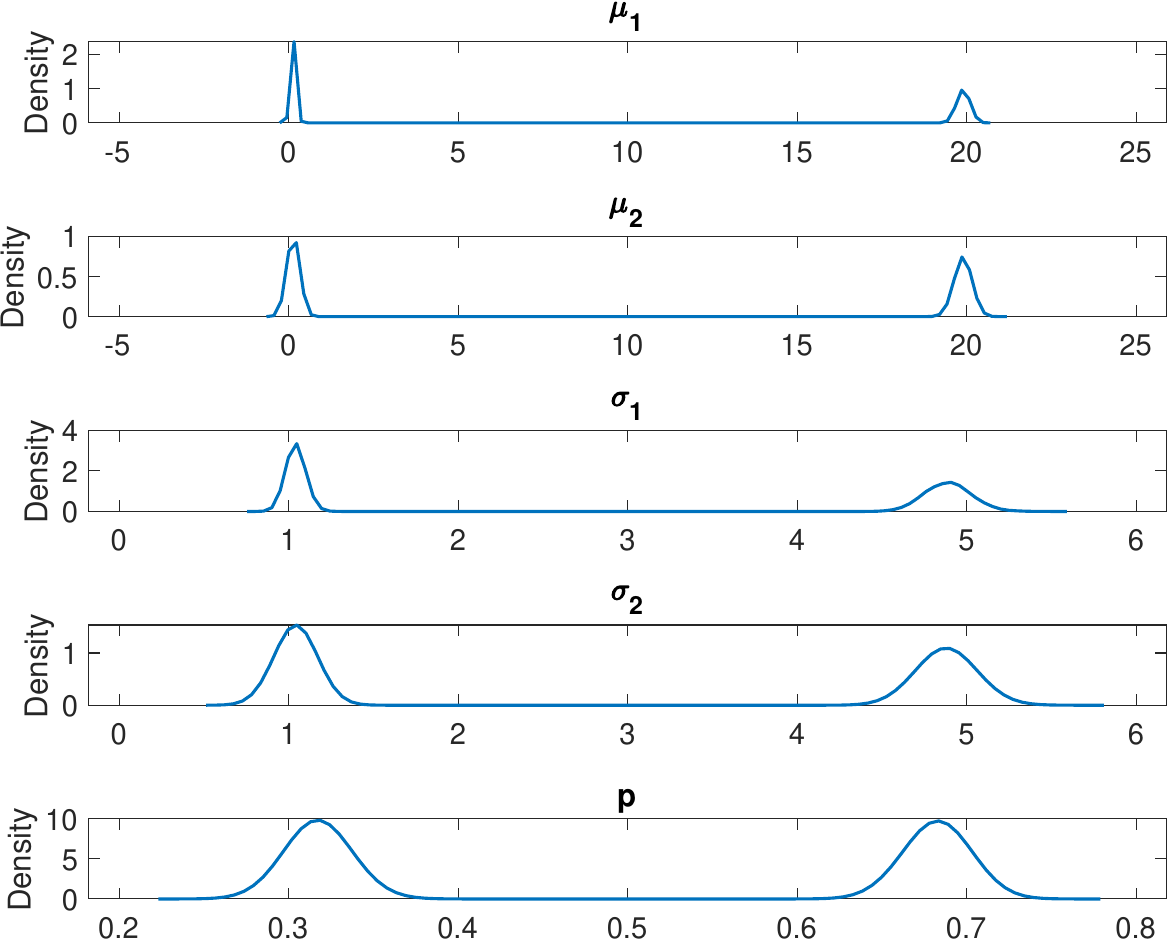}
    \caption{IA--RWM: representative chain}
    \label{fig:density_chain_new}
  \end{subfigure}

  \vspace{0.5em}

  \begin{subfigure}[b]{0.49\textwidth}
    \centering
    \includegraphics[width=\textwidth]{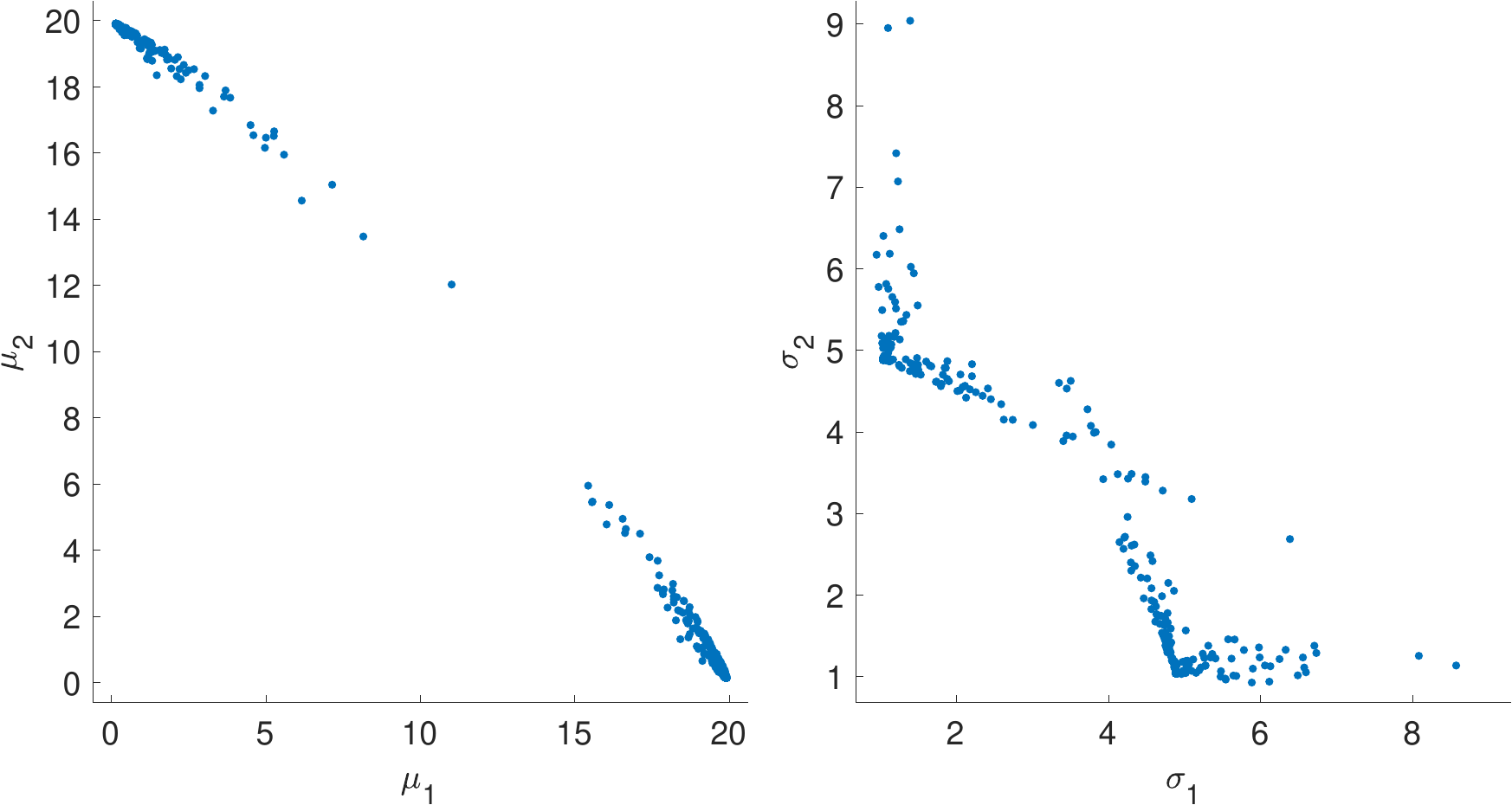}
    \caption{Gaussian RWM: Monte Carlo moments}
    \label{fig:MCMC_moment}
  \end{subfigure}
  \hfill
  \begin{subfigure}[b]{0.49\textwidth}
    \centering
    \includegraphics[width=\textwidth]{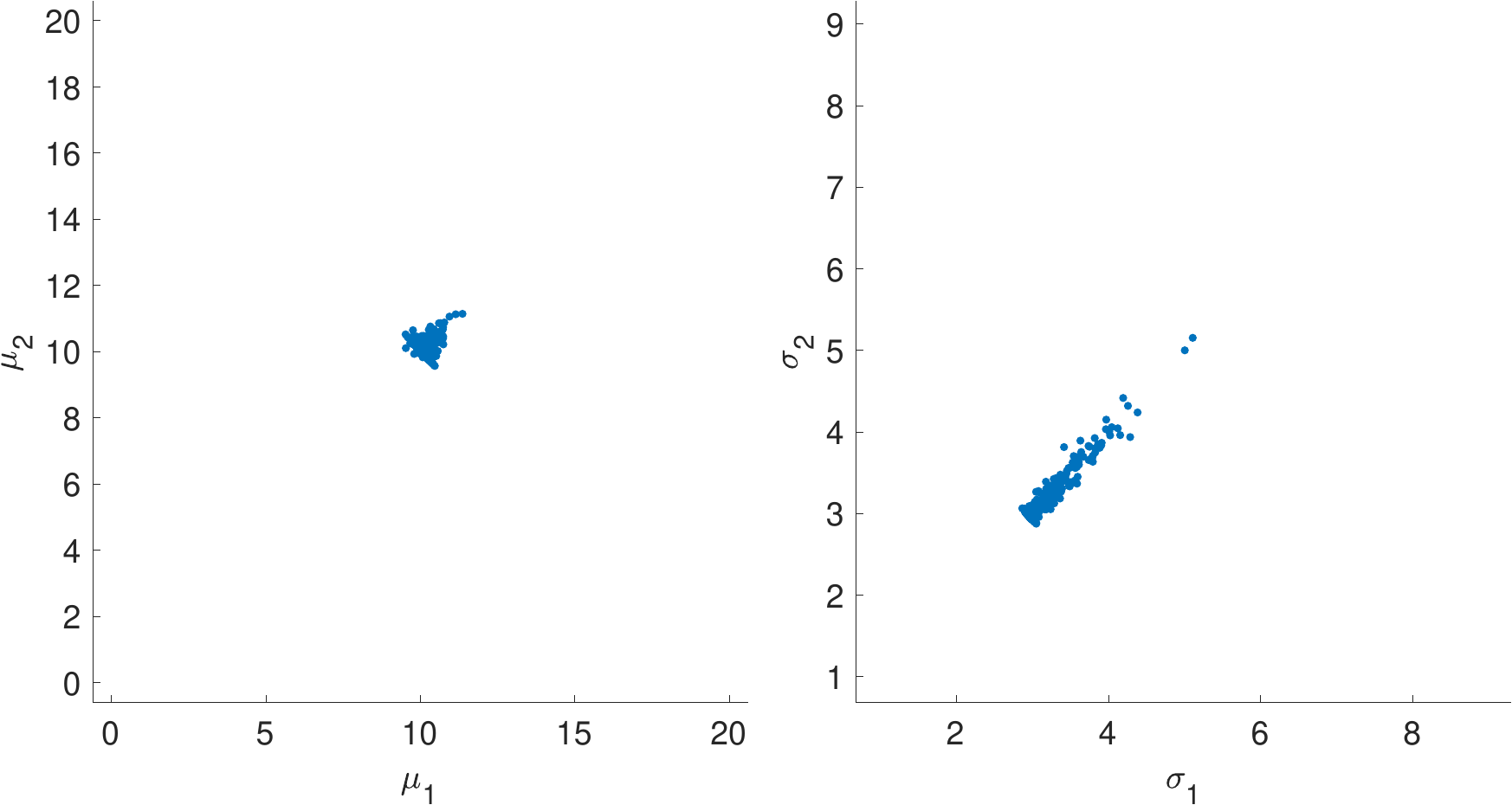}
    \caption{IA--RWM: Monte Carlo moments}
    \label{fig:IAMH_moment}
  \end{subfigure}

  \caption{Comparison of Gaussian RWM and IA--RWM. Panels (a)--(b) show the result of one representative chain, while panels (c)--(d) report the corresponding Monte Carlo distributions of the coordinate-wise sample moments.}
  \label{fig:rwm_ia_comparison}
\end{figure}

The next simulation considers movement along a flat observational equivalence fiber rather than communication between isolated modes. It compares local mutation and tempering with direct conditional sampling along that fiber.

\begin{simulation}[Gaussian mean decomposition]
\label{sim:conditional_gaussian}
For \(t=1,\ldots,N\), let
\[
X_t\mid\mu\stackrel{\mathrm{i.i.d.}}{\sim}
\mathcal N\!\left(\sum_{i=1}^k\mu_i,1\right).
\]
We consider \(k=2\) and \(k=10\). In both cases,
\[
K(\mu_1,\ldots,\mu_k)
=
\left\{(\bar\mu_1,\ldots,\bar\mu_k):
\sum_{i=1}^k\bar\mu_i=\sum_{i=1}^k\mu_i\right\}
\]
is an affine observational equivalence fiber.
\end{simulation}

Both designs use \(N=1{,}000\) observations and the uniform prior on \([-10,10]^k\). For \(k=2\), the true value is \((\mu_1,\mu_2)=(0,2)\). Figure~\ref{fig:cond_gaussian_dim2} compares Gaussian RWM, HMC, and SMC, reporting 10,000 draws. The SMC calculation initializes 10,000 particles from the prior and uses 13 tempering stages with three RWM mutation steps per stage. In this two-dimensional setting, SMC explores the full fiber better.
\begin{figure}[!htbp]
  \centering
  \includegraphics[width=0.8\textwidth]{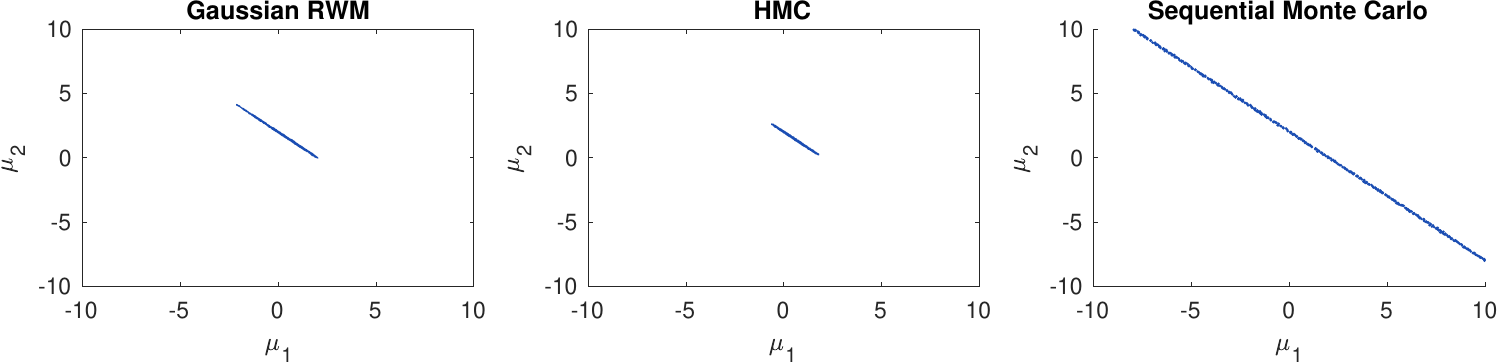}
  \caption{Draws for \(\mu_1\) and \(\mu_2\) in the \(k=2\) Gaussian-mean design.}
  \label{fig:cond_gaussian_dim2}
\end{figure}

For \(k=10\), we compare SMC using local RWM mutation \citep{del2006sequential,herbst2014sequential} with conditional batch augmentation of RWM, a variant that supplements each local-chain draw with a batch of conditional draws along its observational equivalence fiber (Appendix~\ref{alg:batch_aug}). The true parameter values are \(\mu_1=10\) and \(\mu_i=0\) for \(i\neq 1\).
The SMC implementation uses \(10^8\) particles initialized from a Gaussian with coordinate standard deviation \(0.5\), centered at the equal decomposition of the sample mean and truncated to the prior box. It uses 15 tempering stages with three RWM mutation steps per stage. Batch augmentation uses 10,000 local-chain draws and generates 10,000 conditional draws from each.\footnote{We use this variant because exact conditional draws are cheap in this example and require no further likelihood evaluations. Generating them in batches directly exploits this computational advantage, without adding stage-specific teleport updates across the SMC particle population. In the later SVMA applications, direct conditional sampling is less tractable. Teleport moves are instead interleaved with NUTS to improve exploration at little additional cost.} 

Figures~\ref{fig:Particle_SMC_IAMHMC} and~\ref{fig:KernelDensity_SMC_IAMHMC} show the two-dimensional projection and marginal densities, respectively. Under this localized initialization, even \(10^8\) SMC particles remain concentrated near the starting region, whereas conditionally augmented RWM draws span the displayed \([-10,10]^2\) projection. RWM and conditional augmentation took approximately 189 seconds, compared with 1,615 seconds for SMC on an Intel Xeon Gold 6246R CPU with 128 GB RAM.\footnote{These wall-clock times are specific to the reported MATLAB implementations and hardware.} The contrast shows the computational value of identification information: conditional augmentation samples alternative decompositions of the same total mean directly, rather than relying on local updates to move gradually between them.

\begin{figure}[!htbp]
    \centering

    \begin{subfigure}[t]{0.48\textwidth}
        \centering
        \includegraphics[width=\textwidth]{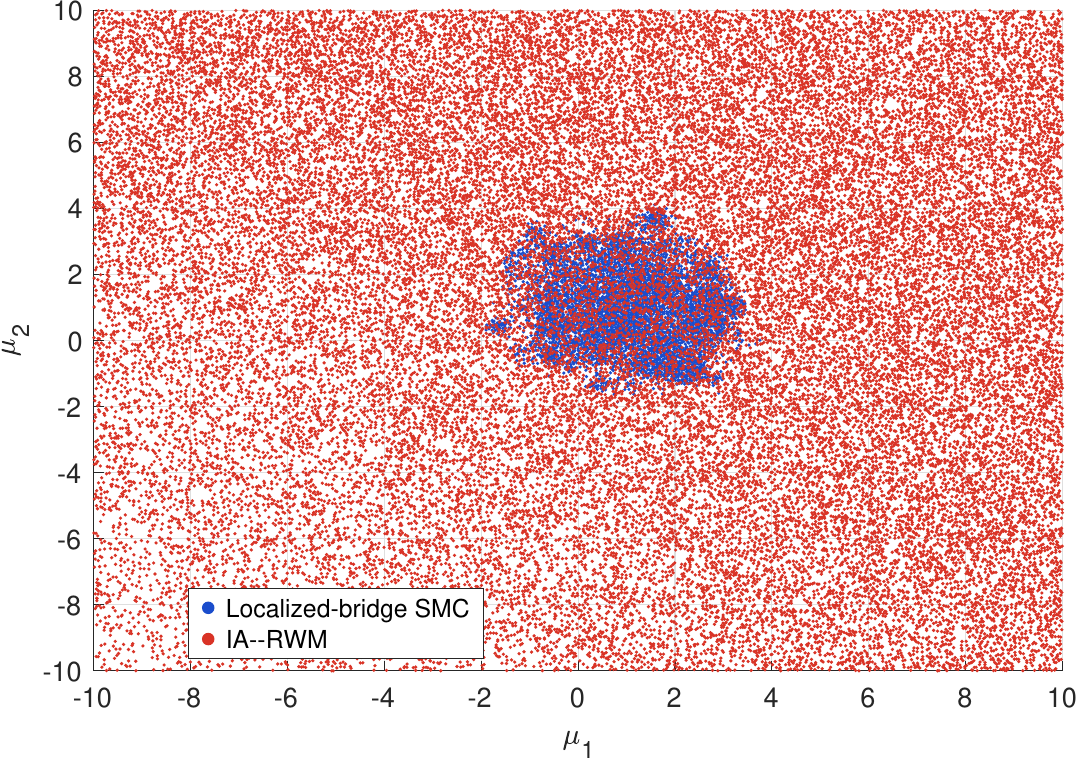}
        \caption{Projected sample draws in \((\mu_1,\mu_2)\).}
        \label{fig:Particle_SMC_IAMHMC}
    \end{subfigure}
    \hfill
    \begin{subfigure}[t]{0.48\textwidth}
        \centering
        \includegraphics[width=\textwidth]{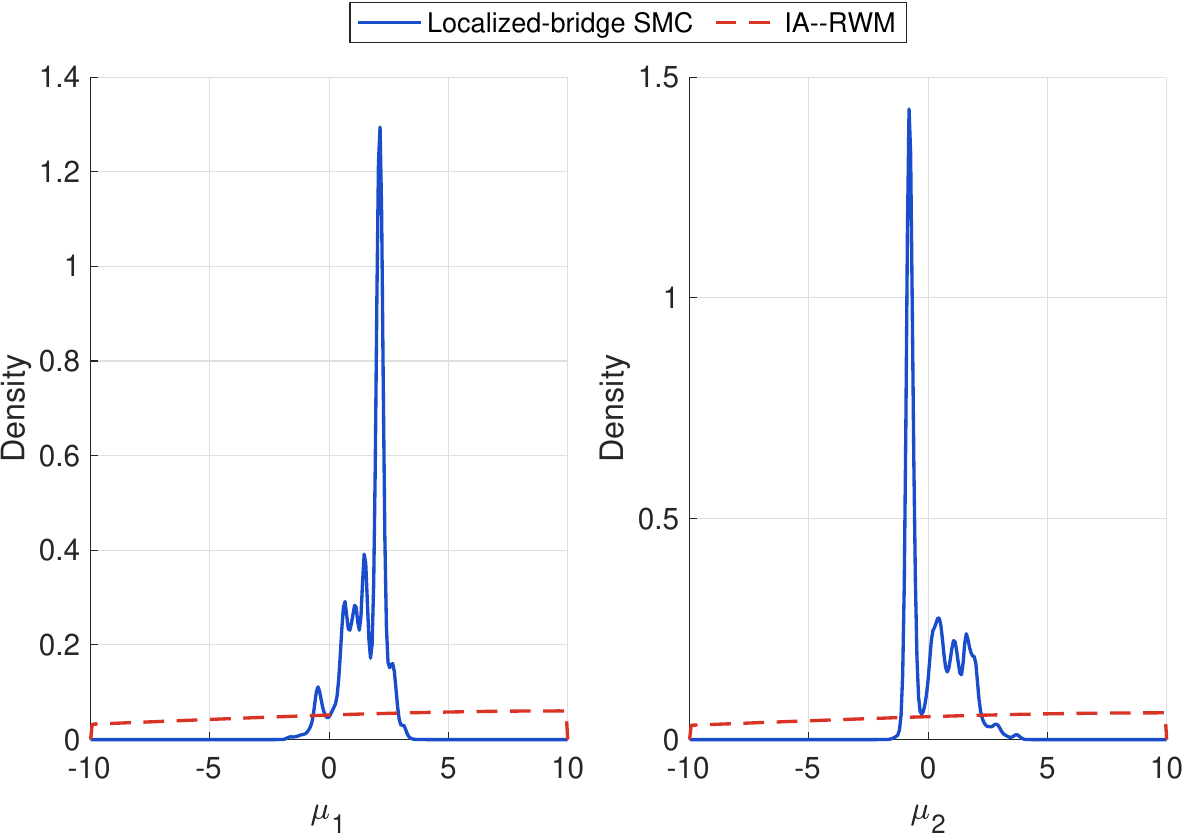}
        \caption{Marginal kernel densities for \(\mu_1\) and \(\mu_2\)}
        \label{fig:KernelDensity_SMC_IAMHMC}
    \end{subfigure}

    \caption{Comparison of SMC and conditionally augmented RWM in the two-dimensional projection.}
    \label{fig:SMC_IAMHMC_comparison}
\end{figure}
\FloatBarrier
\section{SVMA Applications}
\label{sec: application}
Observationally equivalent SVMA representations imply the same distribution of the data but can assign different impulse responses to the structural shocks. Incomplete exploration can therefore affect the economic conclusions drawn from posterior simulation without an obvious difference in fit to the data. We examine this issue under weak and strong priors in two SVMA exercises based on \citet{plagborg2019bayesian}.  The SVMA(1) illustration makes these moves transparent, while the full \(q=16\) application evaluates their practical value in a 153-parameter empirical model. The comparison shows how prior concentration affects the gains from identification-aware moves.

An \(n\)-variable SVMA writes the data as
\[
y_t=\sum_{\ell=0}^{q}\Theta_\ell\varepsilon_{t-\ell},
\qquad
\varepsilon_t\stackrel{\mathrm{i.i.d.}}{\sim}\mathcal N(0,\Sigma),
\qquad
\Sigma=\operatorname{diag}(\sigma_1^2,\ldots,\sigma_n^2).
\]
Thus \(\sigma_j>0\) is the standard deviation, or scale, of structural shock \(j\). For parameterizations with nonsingular impact matrices, \citet[Proposition~2]{plagborg2019bayesian}, building on \citet{lippi1994var}, shows that SVMA observational equivalence is captured by orthogonal rotations and Blaschke root flips. Root flips enter the SVMA(1) comparison and the initialization search described in Online Appendix~\ref{append:empirical}. The \(q=16\) proposal uses rotations only.

Both exercises use an informative and a bounded-uniform prior. The informative prior follows \citet{plagborg2019bayesian}. For each \((i,j)\), the prior over \(\Theta_{ij,0:q}\) is Gaussian, with prior means guided by the log-linearized sticky-price dynamic stochastic general equilibrium (DSGE) impulse-response functions (IRFs) of \citet{sims2012news} and with smoothness across horizons. Shock scales satisfy
\[
\log \sigma_j\sim \mathcal N(\mu_j^\sigma,(\tau_j^\sigma)^2),
\]
with diffuse scale priors. The SVMA(1) and \(q=16\) applications use this same prior form with separately calibrated hyperparameters. The bounded-uniform prior is
\[
\Theta_{ij,\ell}\sim \mathrm{Unif}[-B_\ell,B_\ell],
\qquad
\log\sigma_j\sim \mathrm{Unif}[-8,5],
\]
independently across free coordinates. We set \(B_\ell=100\) in the
SVMA(1) illustration and
\[
B_\ell=\frac{100}{\max\{1,\ell-1\}^{2}},
\qquad \ell=0,\ldots,16,
\]
in the full application.\footnote{The lag-dependent box imposes shrinkage through the support: early responses remain broadly admissible, while increasingly remote responses must be smaller.}
Both priors are specified on the free coordinates after unit-impact normalization. Uniformity in these coordinates does not imply uniformity over observationally equivalent rotations.\footnote{We impose unit-impact normalization \(\Theta_{jj,0}=1\) for \(j=1,\ldots,n\). In sign-restricted SVARs, \citet{arias2025uniform} show that, conditional on the reduced-form parameters, Haar-uniform rotations induce a uniform joint distribution with respect to the induced volume measure on admissible impulse responses, although individual marginals need not be uniform.} The teleport update must account for the conditional measure induced
by this normalization. Online Appendix~\ref{append:empirical} gives the
corresponding change of measure.

We compare NUTS with IA--NUTS and, where indicated, IA--RWM. For both models, we assess coefficient exploration by comparing sampled coefficient distributions with numerical conditional-prior references. These references integrate rotations and root representations for SVMA(1), but only rotations within each selected draw's current root representation for \(q=16\). Each reference averages over conditioning values visited by its own sampler. We summarize the discrepancies using Wasserstein (W1) distances between sampled and reference coefficient distributions, normalized by the reference interquartile range, and coefficient-rank Kolmogorov--Smirnov (KS) distances, which measure departures of conditional percentiles from uniformity. Both diagnostics are averaged across free coefficients.  

Table~\ref{tab:svma_conditional_diagnostics} summarizes the four model--prior comparisons. Traces, marginal densities, autocorrelation-based effective sample size (ESS), and transition costs provide complementary information. We also compare autocovariance distributions across samplers. Because rotations and root flips leave autocovariances unchanged, this comparison separately checks agreement in likelihood-relevant quantities. Online Appendix~\ref{append:empirical_diagnostics} gives the detailed diagnostic definitions and supplementary checks.

\subsection{SVMA(1) Illustration}
\label{subsec:svma1_warmup}

The illustration sets \(n=3\) and \(q=1\). Its scaled transfer function and
spectral density are
\[
\begin{aligned}
\Psi(z)&=(\Theta_0+\Theta_1z)\operatorname{diag}(\sigma),\\
f_{yy}(\omega)&=\Psi(e^{-i\omega})\Psi(e^{-i\omega})^*,
\end{aligned}
\]
where \(^*\) denotes conjugate transpose. Orthogonal rotations and Blaschke root flips leave \(f_{yy}\), and hence the Gaussian and Whittle likelihoods, unchanged while changing the structural impulse responses. The sampled marginals of \(\Theta_0\) and \(\Theta_1\) show how these moves affect exploration.

The bounded-uniform specification weakens direct prior concentration on structural coefficients, while its bounds still restrict the feasible equivalent representations.

The bounded-uniform experiment also shows how unit-impact normalization translates this equivalence geometry into reported impulse responses. For a fixed scaled transfer function, write \(\Psi_\ell=\Theta_\ell\operatorname{diag}(\sigma)\).
Let \(Q\) be Haar-uniform on \(\mathrm{SO}(3)\), the group of \(3\times3\) orthogonal matrices with determinant one. Write \(q_j\) for its \(j\)th column and \(e_i\) for the \(i\)th coordinate unit vector. Then
\[
\Theta_{ij,\ell}^{Q}
=
\frac{e_i'\Psi_\ell q_j}
{e_j'\Psi_0 q_j}.
\]
Writing \(q_j=Z/\lVert Z\rVert\) with \(Z\sim\mathcal N(0,I_n)\) shows that, when the numerator and denominator are not proportional, this ratio has a location--scale Cauchy distribution, with its center and scale determined by \(\Psi\). Likelihood-equivalent rotations can therefore produce broad changes in reported impulse responses. This describes the geometry available under Haar rotations, not the conditional prior itself: the teleport update must also account for the prior support and the change of measure induced by normalization.

Figures~\ref{fig:svma1_informative} and~\ref{fig:svma1_uniform} compare NUTS with rotation-only and rotation-plus-root-flip IA--NUTS under the informative and bounded-uniform priors, respectively. Under the informative prior, the three samplers produce closely overlapping impulse-response distributions, consistent with the small coefficient W1 distances \(0.068\), \(0.058\), and \(0.074\).
\begin{figure}[!htbp]
  \centering
  \includegraphics[width=\textwidth]{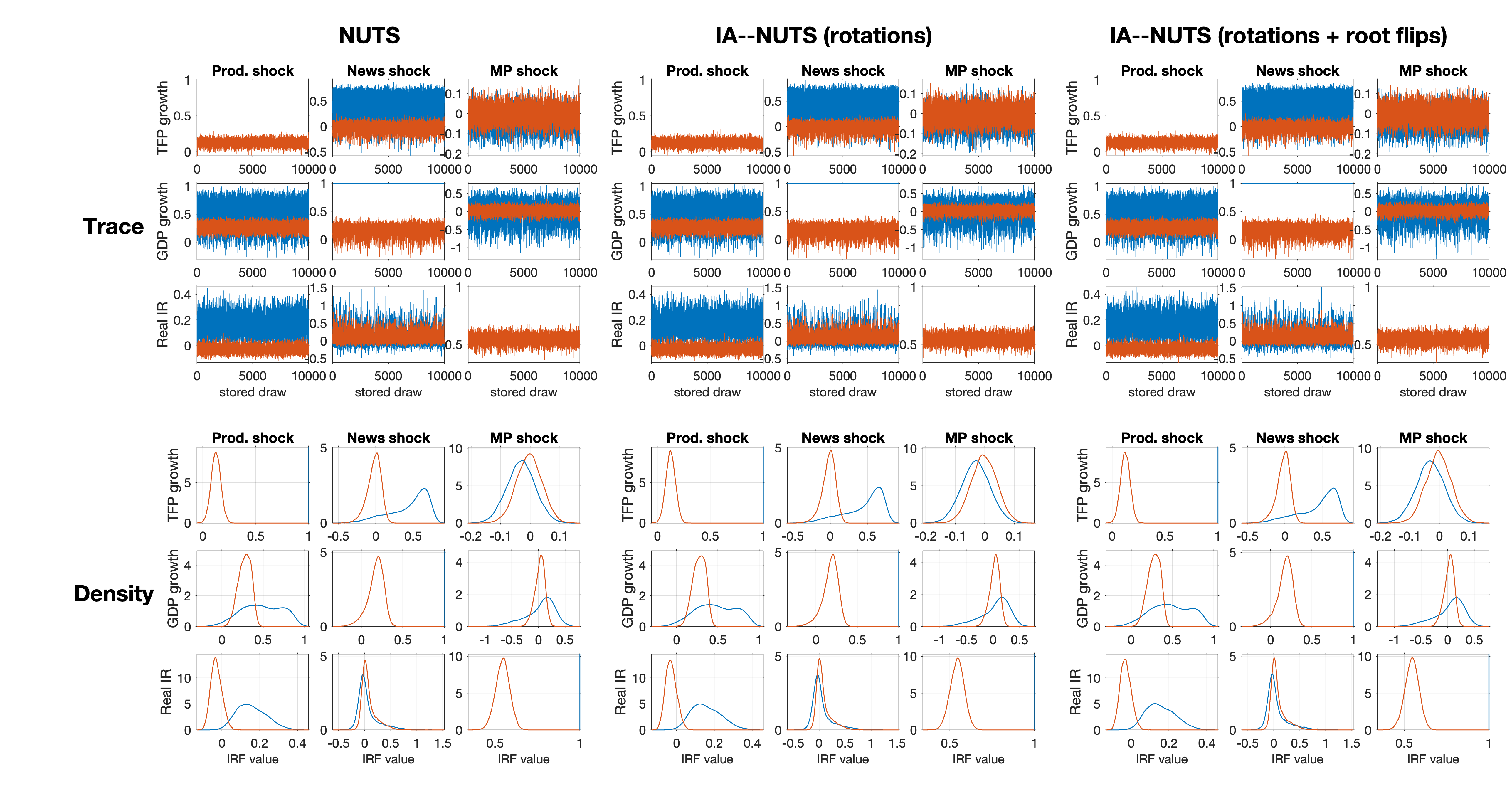}
\caption{SVMA(1) under the informative prior. Columns compare NUTS, rotation-only IA--NUTS, and IA--NUTS with rotations and root flips. Blue and orange denote lags \(0\) and \(1\), respectively.}
  \label{fig:svma1_informative}
\end{figure}

Under the bounded-uniform prior, NUTS remains concentrated in an
asymmetric local region.\footnote{As an additional benchmark, we ran the Metropolis-adjusted Langevin algorithm (MALA), a widely used gradient-based local sampler. Under the informative prior, MALA gives
results similar to NUTS and the identification-aware samplers. Under the bounded-uniform prior, its traces show persistent drift, and its estimated marginal densities are substantially less smooth than those from NUTS. MALA also does not explore the wider range reached by the identification-aware samplers. These results show that limited exploration under the weak prior is not specific to NUTS.} Rotation-only IA--NUTS already broadens the sampled impulse-response distributions, while adding root flips gives the closest agreement with the reference in this illustration. The coefficient diagnostics quantify these gains: mean coefficient W1 falls from \(0.530\) for NUTS to \(0.277\) with rotations and \(0.070\) with rotations and root flips; mean coefficient-rank KS falls from \(0.322\) to \(0.142\) and \(0.092\), respectively (Table~\ref{tab:svma_conditional_diagnostics}). The root-count discrepancy is \(0.620\) for NUTS, \(0.616\) with rotations, and \(0.021\) when root flips are included. The larger \(q=16\) exercise uses rotations only.

\begin{figure}[!htbp]
  \centering
  \includegraphics[width=\textwidth]{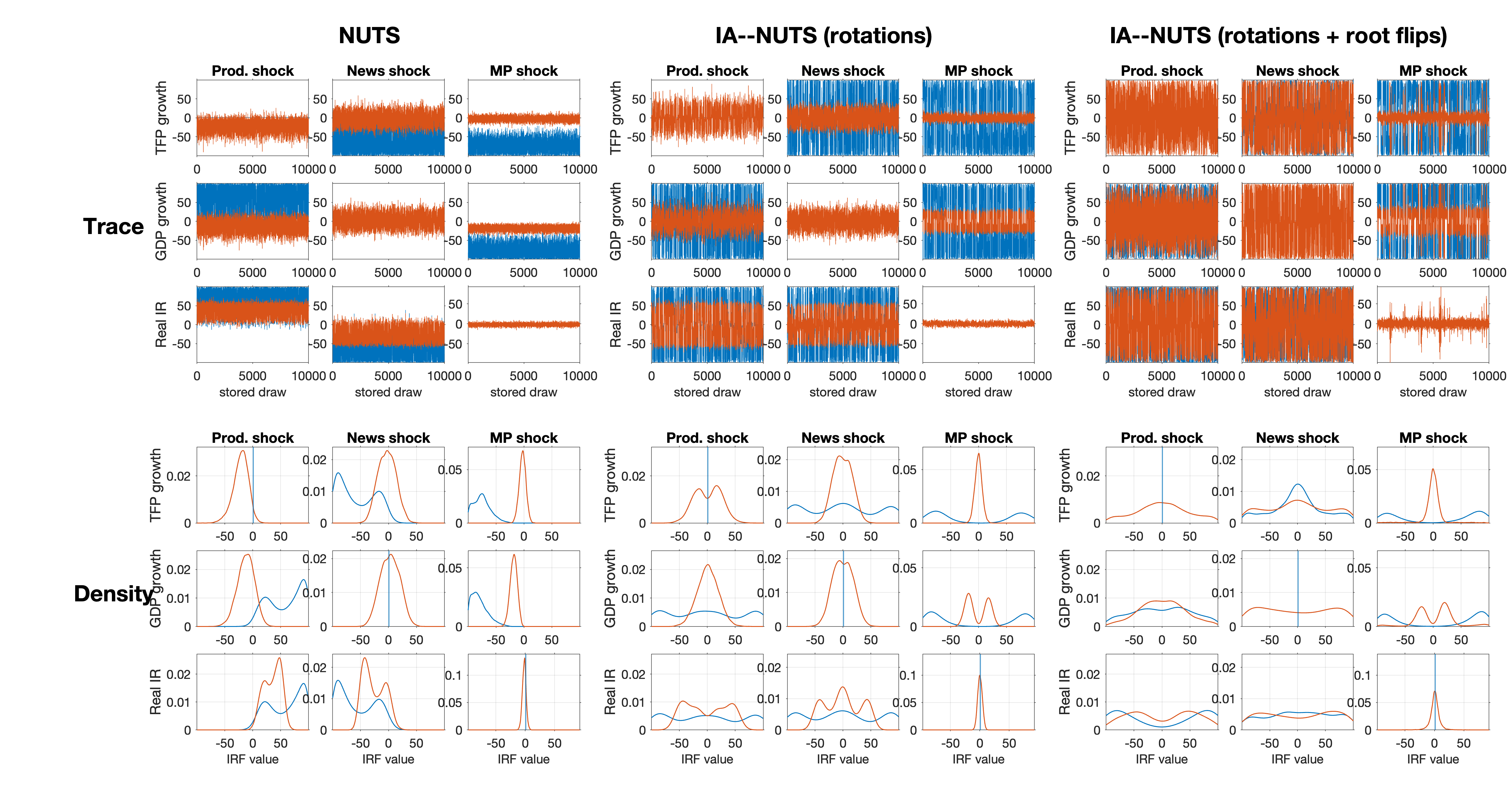}
\caption{SVMA(1) under the bounded-uniform prior. Columns compare NUTS, rotation-only IA--NUTS, and IA--NUTS with rotations and root flips. Blue and orange denote lags \(0\) and \(1\), respectively.}
  \label{fig:svma1_uniform}
\end{figure}

The contrast across priors shows that identification-aware moves matter most when the prior leaves substantial mass across equivalent representations.
 
\subsection{Application: News Shocks and Business Cycles}
\label{subsec:svma_application}

We next apply the samplers to the \(q=16\) news-shock model of
\citet{plagborg2019bayesian}. With \(n=3\) and unit-impact normalization, the model has 153 free parameters. The data are quarterly U.S. total factor productivity (TFP) growth, real gross domestic product (GDP) growth, and the real interest rate, measured as the effective federal-funds rate less contemporaneous GDP-deflator inflation. The sample spans 1954Q3--2007Q4 \((N=213)\). The informative specification uses prior response patterns associated with unanticipated productivity, technology news, and monetary-policy shocks. The bounded-uniform specification provides a contrasting structural prior. 

Our comparison asks whether conventional sampling adequately represents the structural uncertainty under each stated prior. It is not a replication-based test of the published paper's numerical conclusions.

The \(q=16\) proposal applies orthogonal rotations to the scaled transfer function and maps the result back to unit-impact coordinates. Each rotation preserves the spectral density and hence the likelihood used in the application. Online Appendix~\ref{append:empirical} gives the normalization and rotation map.

The main comparison is between NUTS and IA--NUTS. We also include IA--RWM, which uses the same rotation proposal as IA--NUTS but replaces NUTS with blocked random-walk updates. This comparison helps distinguish the contribution of the model-implied nonlocal move from that of the local kernel in 153 dimensions.

\begin{table}[!htb]
\centering
\small
\caption{Conditional calibration and dependence in the SVMA exercises}
\label{tab:svma_conditional_diagnostics}
\begin{tabular}{lrrrr}
\toprule
Sampler & Coefficient W1 & Coefficient KS & Score KS & Median ESS \\
\midrule
\multicolumn{5}{l}{\textit{SVMA(1), bounded uniform}}\\
NUTS                         & 0.530 & 0.322 & 0.221 & 2106 \\
IA--NUTS, rotations          & 0.277 & 0.142 & 0.280 & 519 \\
IA--NUTS, rotations and roots & 0.070 & 0.092 & 0.065 & 599 \\
\addlinespace
\multicolumn{5}{l}{\textit{SVMA(1), informative Gaussian}}\\
NUTS                         & 0.068 & 0.098 & 0.102 & 2818 \\
IA--NUTS, rotations          & 0.058 & 0.091 & 0.058 & 2809 \\
IA--NUTS, rotations and roots & 0.074 & 0.103 & 0.112 & 2692 \\
\addlinespace
\multicolumn{5}{l}{\textit{\(q=16\), bounded uniform}}\\
NUTS                         & 0.125 & 0.148 & 0.098 & 3309 \\
IA--RWM                      & 0.032 & 0.093 & 0.088 & 2001 \\
IA--NUTS                     & 0.033 & 0.093 & 0.062 & 2780 \\
\addlinespace
\multicolumn{5}{l}{\textit{\(q=16\), informative Gaussian}}\\
NUTS                         & 0.027 & 0.084 & 0.101 & 2979 \\
IA--RWM                      & 0.029 & 0.090 & 0.049 & 122 \\
IA--NUTS                     & 0.026 & 0.079 & 0.068 & 2814 \\
\bottomrule
\end{tabular}
\begin{flushleft}
\footnotesize Notes: W1 and coefficient KS assess agreement with numerical conditional-prior references. Score KS applies the conditional-percentile check to the conditional log density. Lower values indicate closer agreement. Median ESS is across free parameters from 10,000 stored draws. 1 and KS use 100 evenly spaced draws per sampler, with 128 conditional completions per selected draw. Both \(q=16\) IA samplers use rotations only.
\end{flushleft}
\end{table}
\paragraph{Informative prior.}
NUTS and IA--NUTS give broadly similar coefficient summaries; see Figure~\ref{fig:append_pm_irf_trace_density} in the Online Appendix.
Table~\ref{tab:svma_conditional_diagnostics} reports mean coefficient W1 values of \(0.027\) and \(0.026\), respectively, providing little evidence of an important additional gain from teleportation. IA--RWM also has a small coefficient discrepancy, but its estimated ESS is substantially lower than those of the two
NUTS-based samplers.

\paragraph{Bounded-uniform prior.}

IA--NUTS substantially improves agreement with the coefficient references at little additional computational cost. Table~\ref{tab:svma_conditional_diagnostics} reports a mean coefficient W1 of \(0.125\) for NUTS and \(0.033\) for IA--NUTS, a reduction of approximately 74 percent. Mean coefficient-rank KS falls from \(0.148\) to \(0.093\). IA--RWM delivers similar coefficient improvements, so the gain is not confined to a NUTS-based implementation.  These gains are not reflected in median autocorrelation ESS, which is lower for IA--NUTS than for NUTS.

Teleport updates account for 2.1 percent of transition time. Of 1,971 attempts, 1,144 (58 percent) change the state, at an average cost of \(0.054\) seconds per attempt.\footnote{These timings cover local and teleport updates, including discarded initial iterations, but exclude preparation, checkpoint output, and diagnostics.}

Figure~\ref{fig:append_uniform_irf_density} shows where the coefficient improvement appears. Several marginal distributions, particularly their sign allocation at impact, differ between NUTS and the identification-aware chains, while other coefficients remain similar. The log shock scales are much closer: the NUTS and IA--NUTS medians differ by less than \(0.01\) for every shock. 

\begin{figure}[!htbp]
  \centering
  \includegraphics[width=\textwidth]{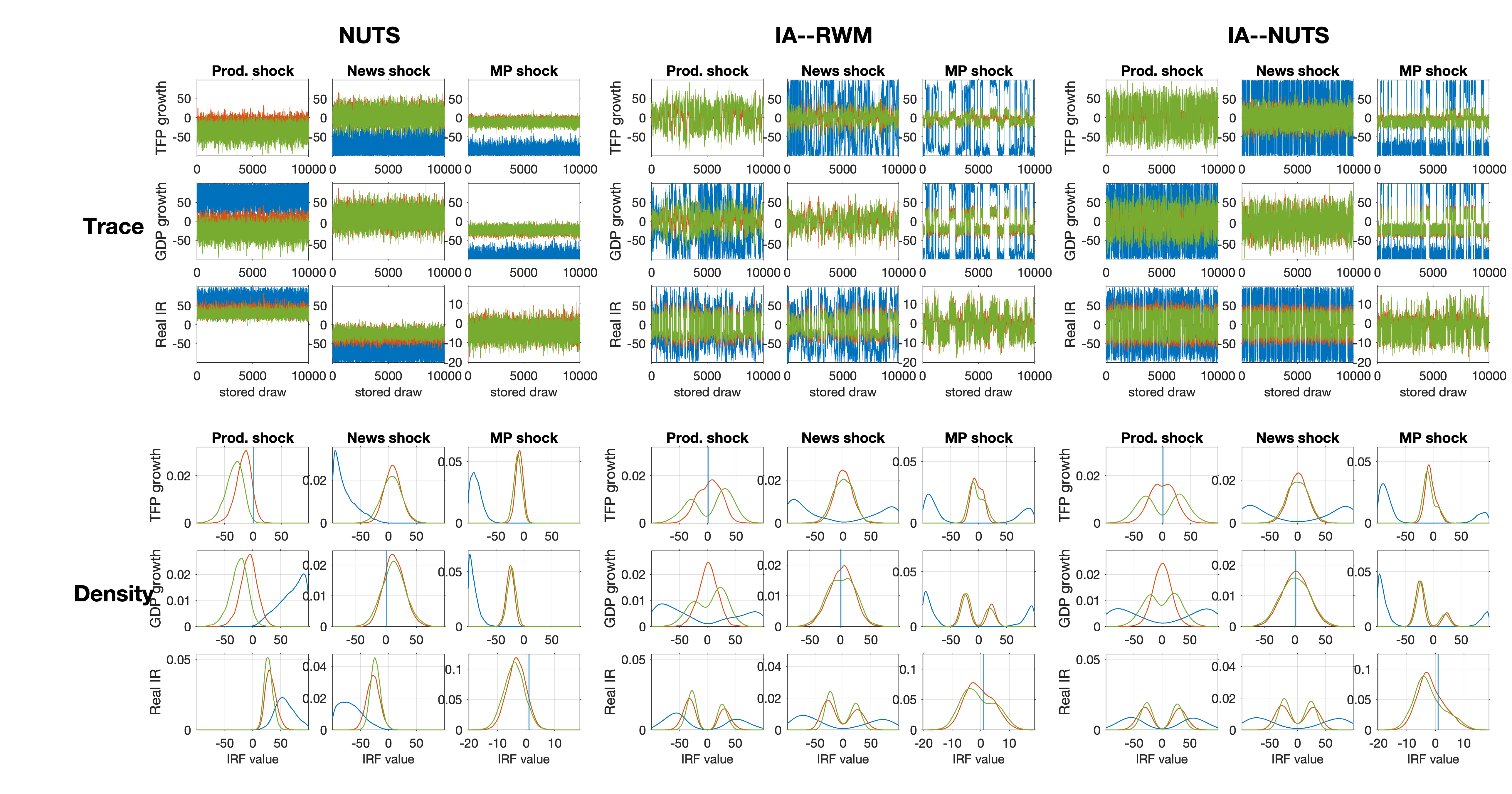}
\caption{Bounded-uniform \(q=16\) normalized coefficients.
Columns correspond to NUTS, IA--RWM, and IA--NUTS.
Blue, orange, and green denote lags \(0\), \(1\), and \(2\),
respectively.}
  \label{fig:append_uniform_irf_density}
\end{figure}

As an auxiliary comparison, Table~\ref{tab:svma_uniform_crosscred}
uses each chain's empirical fifth and ninety-fifth percentiles. Median cross-chain inclusion is close to 90 percent in both directions, but this conceals differences for some coefficients: the lower decile of IA--NUTS inclusion in NUTS intervals is \(0.512\). These are comparisons between saved chains, not coverage probabilities relative to a known posterior.

\begin{table}[!htbp]
\centering
\small
\caption{Empirical 90 percent coefficient intervals:
bounded-uniform \(q=16\)}
\label{tab:svma_uniform_crosscred}
\begin{tabular}{lccc}
\toprule
 & Median & 10th pct. & 90th pct. \\
\midrule
IA--NUTS draws in NUTS interval & 0.891 & 0.512 & 0.903 \\
NUTS draws in IA--NUTS interval & 0.901 & 0.895 & 0.925 \\
IA--NUTS/NUTS interval width   & 1.011 & 0.997 & 2.206 \\
\bottomrule
\end{tabular}
\begin{flushleft}
\footnotesize Notes: For each of the 150 free normalized coefficients, we use all 10,000 stored draws per chain to form an interval between the empirical fifth and ninety-fifth percentiles. The first two rows report, in each direction, the fraction of draws from one chain that fall inside the other chain's interval. The last row reports the IA--NUTS-to-NUTS interval-width ratio. The columns give the median, 10th percentile, and 90th percentile of these quantities across the 150 coefficients.
\end{flushleft}
\end{table}

\section{Conclusions and Discussion}
This paper turns the results of identification analysis into a general framework for improving Bayesian computation. The proposed samplers exploit model-implied observational equivalence to overcome persistent mixing difficulties through nonlocal moves that preserve the target posterior. Our spectral-gap comparisons provide theoretical guarantees for these gains in two settings:  separated modes and extended likelihood-flat regions.

The prior comparisons also have implications for sensitivity analysis. Weakening a prior can give greater posterior weight to equivalent structural representations that local samplers explore poorly. Apparent robustness to prior choice may then reflect incomplete exploration rather than genuinely stable posterior conclusions. Identification-aware moves can reveal the effect of relaxing the prior by reaching equivalent representations that local updates may leave unexplored.

The framework opens further opportunities to exploit identification information in Bayesian computation. These include constructing other model-implied nonlocal proposals outside the OE set, and developing proposals based on approximate equivalence in weakly identified models, with acceptance corrections preserving the target posterior. Further computational gains may come from optimizing how often teleports are attempted and how many candidate states they consider.

\FloatBarrier
\clearpage
\appendix
\section{Teleport Implementation and IA Variants}
This section gives a Multiple-Try Metropolis implementation of the teleport step and additional IA constructions. The hybrid transition used in the main text is stated in Algorithm~\ref{alg:hybrid_ia}.
 \subsection{Multiple-Try Metropolis}
\label{append:mtm_teleport}

When a direct draw from \(T(\theta,\cdot)\) is unavailable, a Multiple-Try Metropolis (MTM) update within \(\mathcal C(\theta)\) can be used instead. Given the current state \(\theta\), set \(z=\chi(\theta)\). For almost every \(z\) under the marginal law of \(\chi\), the conditional law \(\pi(\,\cdot\mid\chi=z)\) is supported on
$\{u\in\Omega:\chi(u)=z\}$.
At any remaining conditioning values, define the update to leave the current state unchanged.

For each \(z\) to which the update applies, choose a measure \(\nu_z\) on \(\{u:\chi(u)=z\}\), such as counting measure for a finite set, Hausdorff measure for a smooth lower-dimensional set, or the pushforward of Haar probability under a specified group parameterization. Suppose \(\pi(\,\cdot\mid\chi=z)\) has density
\[
p_z(u)
=
\frac{d\pi(\,\cdot\mid\chi=z)}{d\nu_z}(u),
\qquad \chi(u)=z.
\]
This density is needed only up to a normalizing constant. For \(v\) with \(\chi(v)=z\), let \(g_z(u\mid v)\) be a proposal density with respect to \(\nu_z\). Assume that
\[
g_z(u\mid v)>0
\]
whenever \(p_z(u)>0\), for every \(v\) with \(\chi(v)=z\). Let \(\lambda_z(u,v)=\lambda_z(v,u)>0\) be a symmetric auxiliary function. Use the same measure, target density, proposal rule, and auxiliary function for every current state with the same value of \(z\), and assume that the construction below defines a Markov kernel on \(\Omega\). Fix an integer \(L\ge2\).

An MTM update uses the weight function
\[
\omega_z(u,v)=p_z(u)g_z(v\mid u)\lambda_z(u,v).
\]
The simplest choice is \(\lambda_z\equiv1\). From \(\theta\):
\begin{enumerate}
\item Draw forward candidates and compute their weights:
\[
u_1,\ldots,u_L\stackrel{\mathrm{i.i.d.}}{\sim}g_z(\cdot\mid\theta),
\qquad \omega_\ell=\omega_z(u_\ell,\theta).
\]
\item If \(\sum_\ell\omega_\ell=0\), remain at \(\theta\) and
stop. Otherwise, select \(u^\ast\) with probability proportional
to \(\omega_\ell\).
\item Draw reverse candidates and include the current state:
\[
v_1,\ldots,v_{L-1}\stackrel{\mathrm{i.i.d.}}{\sim}g_z(\cdot\mid u^\ast),
\qquad v_L=\theta.
\]
Set \(\omega_\ell^\ast=\omega_z(v_\ell,u^\ast)\).
\item If \(\sum_\ell\omega_\ell^\ast=0\), remain at \(\theta\).
Otherwise, accept \(u^\ast\) with probability
\[
\alpha(\theta,u^\ast)
=\min\left\{1,\frac{\sum_{\ell=1}^L\omega_\ell}
{\sum_{\ell=1}^L\omega_\ell^\ast}\right\}.
\]
If rejected, remain at \(\theta\).
\end{enumerate}

For each conditioning value to which the update applies, the standard Multiple-Try Metropolis detailed-balance argument shows that the within-set transition is reversible with respect to \(\pi(\,\cdot\mid\chi=z)\) \citep[Theorem~1]{liu2000multiple}. Because the choices above define a Markov kernel on \(\Omega\), integrating these conditional detailed-balance identities over the marginal law of \(\chi\), as in Lemma~\ref{lemma:conditional_sampling}, shows that \(T_{\mathrm{MTM}}\) is \(\pi\)-reversible. If \(P\) is \(\pi\)-invariant, \(T_{\mathrm{MTM}}\) may therefore replace \(T\) in
\[
(1-\varepsilon)P+\varepsilon T_{\mathrm{MTM}}
\]
without changing the invariant posterior. If \(P\) is also \(\pi\)-reversible, Proposition~\ref{prop:reversibility} in the Online Appendix gives reversibility of the random-order kernel
\[
\bar P_{\mathrm{MTM}}
=
\frac12T_{\mathrm{MTM}}P+\frac12PT_{\mathrm{MTM}}.
\]

For an independence proposal \(g_z(u\mid v)=g_z(u)\), one convenient symmetric choice is
\[
\lambda_z(u,v)
=
\frac{1}{g_z(u)g_z(v)}.
\]
Then the selection weights become
\[
\omega_\ell\propto
\frac{p_z(u_\ell)}{g_z(u_\ell)}.
\]
This form is useful when candidates can be drawn independently from \(g_z\).

\subsection{Identification-aware Variants}
\label{append: ia_variants}
For the measurable refinement defined in Section~\ref{subsec: identification_driven}, \(T(\theta,\cdot)\) is supported on \(\mathcal C(\theta)\) for \(\pi\)-almost every \(\theta\) and preserves \(\pi\). The variants below state their additional invariance or reversibility conditions.

\subsubsection{Random Order}
\label{alg:twostep_priorweighted}

The random-order kernel
\[
\bar P=\frac12TP+\frac12PT
\]
is reversible when \(P\) and \(T\) are both \(\pi\)-reversible. The main results use \(\widetilde P_H=(1-\varepsilon)P+\varepsilon T\).  
Proposition~\ref{prop:reversibility} in the Online Appendix records the result.

\subsubsection{MH Composition}
\label{alg:exact_mixtureMH}
Given a teleport kernel \(T\), define the averaged proposal
\[
\tilde q(\theta''\mid\theta)=\int T(\theta,d\theta')\,q(\theta''\mid\theta').
\]
The corresponding Metropolis--Hastings acceptance probability is
\[
\tilde\alpha(\theta,\theta'')
=
\min\left\{
1,\frac{\pi(\theta'')\tilde q(\theta\mid\theta'')}{\pi(\theta)\tilde q(\theta''\mid\theta)}
\right\}.
\]
The resulting kernel is \(\pi\)-reversible by the standard MH detailed-balance argument. In general this kernel is different from \(PT\), \(TP\), \(\bar P\), and the hybrid kernel \(\widetilde P_H\).

Because evaluating \(\tilde q\) requires integrating the local proposal over \(\mathcal C(\theta)\), we treat this construction as supplementary and use
\[
\widetilde P_H=(1-\varepsilon)P+\varepsilon T
\]
in the main theory.

\subsubsection{Batch Augmentation}
\label{alg:batch_aug}

Batch augmentation uses an existing Markov chain to generate additional draws from \(T(\theta_t,\cdot)\). It is especially useful when these draws are cheaper than additional local transitions. Let \(\{\theta_t\}_{t\ge1}\) be an ergodic Markov chain with invariant distribution \(\pi\). For a fixed batch size \(M\ge1\), draw a batch conditional on each chain state,
\[
U_{t,1},\ldots,U_{t,M}
\mid\theta_t
\stackrel{\mathrm{i.i.d.}}{\sim}
T(\theta_t,\cdot),
\]
where \(T\) draws directly from the conditional law in \eqref{eqn:teleport_kernel}. If the base chain is stationary, \(\pi T=\pi\) gives
\[
\Pr(U_{t,j}\in A)
=
\int T(\theta,A)\pi(d\theta)
=
\pi(A).
\]

For the first \(N\) chain states, the corresponding stacked empirical measure is
\[
\hat\Pi_{\mathrm{stack}}
=
\frac{1}{NM}
\sum_{t=1}^N\sum_{j=1}^M\delta_{U_{t,j}}.
\]
For \(g\in L_1(\pi)\),
\[
\int g\,d\hat\Pi_{\mathrm{stack}}
=
\frac{1}{N}\sum_{t=1}^N
\left\{\frac{1}{M}\sum_{j=1}^M g(U_{t,j})\right\}
\longrightarrow
\pi(g)
\]
under the usual ergodic and integrability conditions.

More generally, if \(M_t\ge1\) is chosen before drawing the \(t\)-th batch, use an empirical measure that gives each batch equal weight:
\[
\hat\Pi_{\mathrm{batch}}
=
\frac{1}{N}\sum_{t=1}^N
\left\{\frac{1}{M_t}\sum_{j=1}^{M_t}\delta_{U_{t,j}}\right\}.
\]
Giving each batch equal weight prevents state-dependent batch sizes from reweighting the base chain. Pooling all draws with weights \(1/\sum_t M_t\) need not target \(\pi\) when \(M_t\) depends on \(\theta_t\).

The auxiliary draws are independent conditional on their shared base state, but generally dependent marginally. Successive batches also inherit serial dependence from the base chain. Monte Carlo uncertainty should therefore be assessed from the sequence of batch averages rather than from the raw pooled count. The \(k=10\) exercise is the fixed-size case with \(N=M=10{,}000\).


\section{Proofs and Supporting Results}
\label{append:proof}

This section collects the supporting results and proofs for the invariance and spectral-gap comparisons in Section~\ref{sec: MCMCandMixing}.

\subsection{Supporting Results}
\label{append:supporting_lemmas}

\begin{lemma}[Cheeger's inequality]
\label{lemma: cheeger}

Let \(P\) be a reversible Markov transition kernel with invariant measure \(\pi\), and let \(\gamma(P)\) be the Poincar\'e spectral gap defined in Section~\ref{subsec: convergence}. Define the conductance
\[
\mathbf h_P
=
\inf_{S\subseteq\Omega,\ 0<\pi(S)\le 1/2}
\frac{\int_S P(\theta,S^c)\pi(d\theta)}{\pi(S)}.
\]
Then
\[
\frac{\mathbf h_P^2}{2}\le \gamma(P)\le 2\mathbf h_P.
\]
\end{lemma}
See \citet[Theorem~2.1]{lawler1988bounds} for general state spaces and \citet[Proposition~6]{diaconis1991geometric} for the finite-state version.

\begin{lemma}[State decomposition theorem]
\label{lemma: state_decomp}
Let \(P\) be a \(\pi\)-reversible Markov kernel on a measurable state space \((\Omega,\mathcal F)\). Let \(\{A_1,\ldots,A_m\}\) be a measurable partition of \(\Omega\) with \(\pi(A_i)>0\) for every \(i\). The restricted kernel on \(A_i\) is
\[
P_{A_i}(\theta,B)=P(\theta,B)+\mathbf 1_B(\theta)P(\theta,A_i^c),
\qquad \theta\in A_i,
\quad B\subseteq A_i.
\]
The component chain on \(\{1,\ldots,m\}\) is defined by
\[
P_H(i,j)=\frac{1}{\pi(A_i)}\int_{A_i}P(\theta,A_j)\pi(d\theta),
\qquad i\neq j,
\]
and
\[
P_H(i,i)=1-\sum_{j\neq i}P_H(i,j).
\]
Then \(P_H\) is reversible with stationary weights \((\pi(A_1),\ldots,\pi(A_m))\), each \(P_{A_i}\) is reversible with invariant law \(\pi(\cdot\mid A_i)\), and
\[
\gamma(P)\ge \frac14\,\gamma(P_H)\min_{1\le i\le m}\gamma(P_{A_i}).
\]
If \(P\) is nonnegative definite on \(L_2(\pi)\), the factor \(1/4\) can be replaced by \(1/2\).
\end{lemma}

\begin{proof}
The reversibility statements follow directly from detailed balance. If \(P\)
is nonnegative definite, the general-state operator result in
\citet[Theorem~A.1]{madras2002markov}, together with
\citet[Theorem~3.2]{martin2006disjoint}, gives
\[
\gamma(P)
\ge
\frac12\,\gamma(P_H)
\min_{1\le i\le m}\gamma(P_{A_i}).
\]

For a general reversible \(P\), let \(P_L=(I+P)/2\). This kernel is nonnegative definite. Its component chain and restricted kernels are
\[
(P_L)_H=\frac{I+P_H}{2},
\qquad
(P_L)_{A_i}=\frac{I+P_{A_i}}{2}.
\]
Consequently,
\[
\gamma(P_L)=\frac12\gamma(P),
\qquad
\gamma((P_L)_H)=\frac12\gamma(P_H),
\qquad
\gamma((P_L)_{A_i})=\frac12\gamma(P_{A_i}).
\]
Applying the nonnegative-definite result to \(P_L\) and substituting these identities gives the factor \(1/4\).
\end{proof}

\begin{lemma}[Reversible Doeblin minorization]
\label{lemma:reversible_doeblin_gap}
Let \(P\) be a \(\pi\)-reversible Markov kernel. Suppose there exist \(n_0\in\mathbb N\), \(\eta\in(0,1]\), and a probability measure \(\nu\) such that
\[
P^{n_0}(\theta,\cdot)\ge \eta\nu(\cdot)
\qquad\text{for all }\theta.
\]
Then \(\gamma(P)\ge c(n_0,\eta)>0\), where the constant depends only on
\((n_0,\eta)\).
\end{lemma}

\begin{proof}
Set \(Q=P^{n_0}\). Since \(\pi\) is invariant for \(Q\), integrating the
minorization gives
\(\pi(A)=\int Q(\theta,A)\,\pi(d\theta)\ge\eta\nu(A)\)
for every measurable \(A\). Hence every \(f\in L_2(\pi)\) also belongs
to \(L_2(\nu)\).

Let \(f\in L_2^0(\pi)\). Then
\[
\begin{aligned}
\langle f,(I-Q)f\rangle_\pi
&=
\frac12
\int \pi(d\theta)\,Q(\theta,d\theta')
\bigl\{f(\theta)-f(\theta')\bigr\}^2\\
&\ge
\frac{\eta}{2}
\int \pi(d\theta)\,\nu(d\theta')
\bigl\{f(\theta)-f(\theta')\bigr\}^2\\
&=
\frac{\eta}{2}
\left\{
\pi(f^2)+\nu(f^2)-2\pi(f)\nu(f)
\right\}\\
&=
\frac{\eta}{2}
\left\{
\pi(f^2)+\nu(f^2)
\right\}
\ge
\frac{\eta}{2}\|f\|_2^2.
\end{aligned}
\]
Taking the infimum over all nonzero \(f\in L_2^0(\pi)\), using the
definition of the Poincar\'e gap in
Section~\ref{subsec: convergence}, gives
\(\gamma(P^{n_0})=\gamma(Q)\ge\eta/2\).

Because \(P\) is reversible, the spectral theorem applies to its
\(L_2(\pi)\) operator. The scalar inequality
\(1-\lambda^{n_0}\le n_0(1-\lambda)\) for \(-1\le\lambda\le1\)
therefore implies
\[
\langle f,(I-P^{n_0})f\rangle_\pi
\le
n_0\langle f,(I-P)f\rangle_\pi.
\]
Combining this inequality with the bound above gives, for every
\(f\in L_2^0(\pi)\),
\[
\langle f,(I-P)f\rangle_\pi
\ge
\frac{\eta}{2n_0}\|f\|_2^2.
\]
Taking the infimum once more yields \(\gamma(P)\ge\eta/(2n_0)>0\),
which depends only on \((n_0,\eta)\).
\end{proof}

\begin{lemma}[Conditional sampling on a measurable partition]
\label{lemma:conditional_sampling}
Let \((\Theta,\mathcal F)\) be a standard Borel space, let \(\chi:\Theta\to\mathcal X\) be measurable, and let \(\pi\) be a probability measure on \(\Theta\). Let \(\eta\) be the marginal law of \(\chi\) under \(\pi\), and let \(\pi(\,\cdot\mid\chi=z)\) be a regular conditional law. Thus, for every measurable \(A\subseteq\Theta\),
\[
\pi(A)
=
\int_{\mathcal X}\pi(A\mid\chi=z)\,\eta(dz).
\]
Define
\[
T(\theta,A)
:=
\pi(A\mid\chi=\chi(\theta)).
\]
Then \(T\) is \(\pi\)-reversible.
\end{lemma}

\begin{proof}
For measurable \(A,B\subseteq\Theta\),
\[
\int_A T(\theta,B)\pi(d\theta)
=
\int_{\mathcal X}
\pi(A\mid\chi=z)\pi(B\mid\chi=z)\,\eta(dz),
\]
which is symmetric in \(A\) and \(B\).
\end{proof}

\subsection{Proofs of Main Results}
\label{append:proof_main}
\noindent\textbf{Proposition~\ref{prop:stationary}.}
\begin{proof}
Use the convention
\[
(PT)(\theta,A)
=
\int P(\theta,dz)T(z,A).
\]
If \(P\) and \(T\) preserve \(\pi\), then
\[
\begin{aligned}
\int (PT)(\theta,A)\pi(d\theta)
&=
\int\!\!\int T(z,A)P(\theta,dz)\pi(d\theta) \\
&=
\int T(z,A)\pi(dz)
=
\pi(A).
\end{aligned}
\]
The same argument, with \(P\) and \(T\) interchanged, applies to \(TP\). For the hybrid kernel,
\[
\int \widetilde P_H(\theta,A)\pi(d\theta)
=
(1-\varepsilon)\int P(\theta,A)\pi(d\theta)
+
\varepsilon\int T(\theta,A)\pi(d\theta)
=
\pi(A).
\]
If \(P\) and \(T\) are \(\pi\)-reversible, then \(\widetilde P_H\) is \(\pi\)-reversible because a convex combination of reversible kernels is reversible.
\end{proof}

\phantomsection
\noindent\textbf{Proposition~\ref{prop:deltaMH}.}
\label{append:proof_deltaMH}
\begin{proof}
Write
\[
Q_P(S,S^c):=\int_S P(\theta,S^c)\pi(d\theta),
\qquad
\Phi_P(S):=\frac{Q_P(S,S^c)}{\pi(S)}.
\]

Since at least one modal region has posterior mass at most \(1/2\), choose an index \(i\) with \(\pi(A_i)\le1/2\), and set \(S=A_i\). By definition of \(d_*:=\min_jd_j\), we have \(d_i\ge d_*\). If \(\operatorname{dist}(\theta,A_i^c)>\delta\), every proposal in \(\Omega\) remains in \(A_i\), while proposals outside \(\Omega\) are rejected. Hence \(P(\theta,A_i^c)=0\), and
\[
Q_P(A_i,A_i^c)
\le
\pi\{\theta\in A_i:\operatorname{dist}(\theta,A_i^c)\le\delta\}.
\]
Set \(r_i:=(d_i-\delta)_+\). The boundary strip is contained in
\(\{\theta\in A_i:\|\theta-\mu_i\|\ge r_i\}\); when \(r_i=0\) the bound below is
trivial, and when \(r_i>0\) it follows from
Assumption~\ref{assump:multimodality}\ref{assump:multimodality1}. Hence
\[
Q_P(A_i,A_i^c)
\le
w_i c_i e^{-\nu_i r_i}
\le
w_i c_{\max}e^{-\nu_{\min}(d_*-\delta)_+}.
\]
Since \(\pi(A_i)=w_i\), the preceding bound and Cheeger's inequality give
\[
\begin{aligned}
\Phi_P(A_i)
&=\frac{Q_P(A_i,A_i^c)}{\pi(A_i)}
\le c_{\max}e^{-\nu_{\min}(d_*-\delta)_+},\\
\gamma(P)
&\le 2\mathbf h_P\le 2\Phi_P(A_i)
\le 2c_{\max}e^{-\nu_{\min}(d_*-\delta)_+}.
\end{aligned}
\]

For the hybrid kernel \(\widetilde P_H=(1-\varepsilon)P+\varepsilon T\),
the state-decomposition lemma gives
\[
\gamma(\widetilde P_H)
\ge
\frac14\,\gamma(P_H^\varepsilon)
\min_i \gamma((\widetilde P_H)_{A_i}).
\]

For \(i\neq j\),
\[
\begin{aligned}
P_H^\varepsilon(i,j)
&=
\frac{1}{\pi(A_i)}
\int_{A_i}\widetilde P_H(\theta,A_j)\pi(d\theta)\\
&\ge
\varepsilon
\frac{1}{\pi(A_i)}
\int_{A_i^\circ}T(\theta,A_j^{\mathrm{int}})\pi(d\theta)
\ge
\varepsilon\varepsilon_1\varepsilon_2
=:\beta.
\end{aligned}
\]
Hence \(\gamma(P_H^\varepsilon)\ge \beta^2/2\) by Cheeger's inequality.

For the restricted within-mode chains, choosing \(P\) for \(n_0\)
consecutive steps and using
Assumption~\ref{assump:multimodality}\ref{assump:multimodality3} gives
\[
\begin{aligned}
(\widetilde P_H)_{A_i}^{\,n_0}(\theta,\cdot)
&\ge (1-\varepsilon)^{n_0}P_{A_i}^{\,n_0}(\theta,\cdot)\\
&\ge (1-\varepsilon)^{n_0}\eta_0m_i(\cdot).
\end{aligned}
\]
Each restricted hybrid chain is reversible, so
Lemma~\ref{lemma:reversible_doeblin_gap} gives a spectral-gap lower bound
\(c_{\mathrm{loc}}>0\) depending only on \((\varepsilon,n_0,\eta_0)\),
not on \(d_*\).
Combining,
\[
\gamma(\widetilde P_H)
\ge
\frac{\beta^2c_{\mathrm{loc}}}{8}
=:c_0>0,
\]
where \(c_0\) is independent of \(d_*\).
\end{proof}

\phantomsection
\noindent\textbf{Proposition~\ref{prop:manifold}.}
\label{append:proof_manifold}
\begin{proof}
We prove the two parts separately.

\medskip
\noindent\textit{Part (1): upper bound for the local RWM kernel.}
For each \(u\in\Phi\), write \(S_u=S\cap\mathcal F(u)\), and let
\[
r(\theta):=
d_{\mathcal F}\bigl(\theta,\partial_{\mathcal F}S_{\phi(\theta)}\bigr),
\qquad \theta\in\Theta.\]
By Assumption~\ref{assump:local_nonid}\ref{assump:local_nonid3}, if \(\|z\|<a_0r(\theta)\), then a proposal \(\theta+z\) cannot cross from \(S\) to \(S^c\). Therefore
\[
P(\theta,S^c)\le \overline G(a_0r(\theta)).
\]
The displayed bound in Assumption~\ref{assump:local_nonid}\ref{assump:local_nonid3} gives
\[
\pi(r(\theta)\le s\mid\phi=u)\le C_S\min\left\{1,\frac{s}{D(u)}\right\}.
\]
With \(Z\sim g\), Fubini's theorem and the preceding bound give
\[
\begin{aligned}
\int_{S_u}\overline G(a_0r(\theta))\pi(d\theta\mid\phi=u)
&=
\mathbb E_Z\!\left[
\pi\bigl(\{\theta\in S_u:\ r(\theta)<\|Z\|/a_0\}\mid\phi=u\bigr)
\right]  \\
&\le
C_S\,\mathbb E_Z\!\left[
\min\left\{1,\frac{\|Z\|}{a_0D(u)}\right\}
\right] \\
&\le
\frac{C_S}{a_0D(u)}\,\mathbb E\|Z\|.
\end{aligned}
\]
Since \(\mathbb E\|Z\|=\int_0^\infty\overline G(t)\,dt<\infty\), this is bounded
by \(C/D(u)\).
Averaging over \(u\),
\[
\int_S P(\theta,S^c)\pi(d\theta)
\le
C\int_\Phi \frac{1}{D(u)}m(du).
\]
If \(\pi(S)>1/2\), replace \(S\) by \(S^c\). The fiberwise boundary is unchanged, and the no-crossing conclusion in Assumption~\ref{assump:local_nonid}\ref{assump:local_nonid3} applies to either side of the cut. The same item bounds the denominator below by \(\min\{c_-,1-c_+\}\). Cheeger's inequality yields
\[
\gamma(P)\le C'\int_\Phi\frac{1}{D(u)}m(du)
=
C'\mathfrak D^{-1}.
\]

\medskip
\noindent\textit{Part (2): lower bound for the hybrid kernel.}
Let
\[
B^\star:=B(s(u^\star),\rho/16),
\qquad
\nu(A):=\frac{\lambda(A\cap B^\star)}{\lambda(B^\star)},
\]
where \(\lambda\) is Lebesgue measure. By
Assumption~\ref{assump:local_nonid}\ref{assump:local_nonid4}, local proposals along the specified polygonal paths have acceptance probability at least \(\alpha_0\). Define
\[
\alpha_0:=p_{\min}/p_{\max}>0,
\qquad q_0:=c_g\lambda(B(0,\rho/16))>0.
\]

Fix any starting point \(\theta\), and let \(u=\phi(\theta)\). Consider the following event over \(n=N_0+2\) steps of the hybrid chain:
\begin{enumerate}
\item the next step applies \(T\);
\item \(T\) lands in \(B(s(u),\rho/16)\cap\mathcal F(u)\), which has probability at least \(\kappa_T\);
\item the next \(N_0+1\) steps apply \(P\), and proposals follow the polygonal path from \(s(u)\) to \(s(u^\star)\), ending with a proposal into \(B^\star\).
\end{enumerate}
Because each path segment has length at most \(\rho/16\), moving from any point in one \(\rho/16\)-ball to the next \(\rho/16\)-ball can be achieved by proposal increments lying in a ball of radius \(\rho/16\) whose center has norm at most \(2\rho/16\). All such increments have norm at most \(3\rho/16<\rho/4\). Thus each local step has proposal probability at least \(q_0\), and the final step minorizes Lebesgue measure on \(B^\star\). Since all these balls lie in the \(\rho/4\)-neighborhood of \(\mathcal K_0\), each acceptance probability is at least \(\alpha_0\). Define
\[
\beta
:=
\varepsilon\,\kappa_T(1-\varepsilon)^{N_0+1}(q_0\alpha_0)^{N_0+1}.
\] These bounds give, for \(n=N_0+2\),
\[
\widetilde P_H^{\,n}(\theta,A)
\ge
\beta\nu(A).
\]

Across targets with increasing fiber diameters, \(\beta\) remains uniformly positive if \(N_0\) is uniformly bounded above and \(\varepsilon\), \(1-\varepsilon\), \(\kappa_T\), \(q_0\), and \(\alpha_0\) are uniformly bounded away from zero. Lemma~\ref{lemma:reversible_doeblin_gap} then gives a uniform positive lower bound on the Poincar\'e spectral gap.
\end{proof}

\FloatBarrier
\clearpage
\bibliographystyle{plainnat}
\bibliography{references}
\clearpage
\section*{Online Appendix}
\addcontentsline{toc}{section}{Online Appendix}
\section{Additional Algorithms}

This section records the fixed-length HMC and Metropolis--Hastings conventions and the SMC benchmark used in the paper. IA implementations and variants are given in Appendix~\ref{append:mtm_teleport} and~\ref{append: ia_variants}.

\begin{algorithm}[Fixed-length HMC]
\phantomsection
\label{alg: HMC}
Given \(\pi(\theta)\propto e^{-U(\theta)}\), a positive-definite mass matrix \(M\), fixed step size \(\eta>0\), fixed leapfrog length \(L\), and initial state \(\theta_0\), repeat for
\(t=1,\ldots,N\):
\begin{enumerate}
\item Draw \(p_0\sim\mathcal N(0,M)\); set
\((\theta,p)=(\theta_{t-1},p_0)\).
\item Make a half momentum step
\(p\leftarrow p-(\eta/2)\nabla U(\theta)\).
For \(j=1,\ldots,L\), update
\(\theta\leftarrow\theta+\eta M^{-1}p\) and, if \(j<L\),
\(p\leftarrow p-\eta\nabla U(\theta)\).
Finish with another half momentum step and set
\((\theta^*,p^*)=(\theta,-p)\).
\item Define \(H(\theta,p)=U(\theta)+p'M^{-1}p/2\) and
\(\Delta H=H(\theta^*,p^*)-H(\theta_{t-1},p_0)\).
Accept \(\theta^*\) with probability \(\min\{1,e^{-\Delta H}\}\);
otherwise retain \(\theta_{t-1}\).
\end{enumerate}
Return \(\{\theta_t\}_{t=1}^N\).
\end{algorithm}
The tuning remains fixed during sampling. NUTS chooses trajectory length through its stopping and proposal-selection rules and commonly tunes step size and mass during warm-up \citep{hoffman2014no}.

\subsection{Metropolis--Hastings}
\begin{algorithm}[Metropolis--Hastings]
\phantomsection
\label{alg: MHMC}
Given target density \(\pi\), proposal \(q\), and initial state
\(\theta_0\), repeat for \(t=1,\ldots,N\): propose
\(\theta'\sim q(\cdot\mid\theta_{t-1})\) and accept with probability
\[
\min\left\{1,
\frac{\pi(\theta')q(\theta_{t-1}\mid\theta')}
     {\pi(\theta_{t-1})q(\theta'\mid\theta_{t-1})}\right\}.
\]
Otherwise retain \(\theta_{t-1}\). Return \(\{\theta_t\}_{t=1}^N\).
\end{algorithm}
\subsection{Sequential Monte Carlo}
\label{appendix: SMC}
We use geometric tempering \citep{del2006sequential},
which includes prior-to-posterior likelihood tempering
\citep{herbst2014sequential}.
\begin{algorithm}[SMC]
\phantomsection
\label{alg: SMC}
Given the posterior density \(\pi\), known up to a constant, an initial density \(q_0\) positive on the posterior support, and \(0=\lambda_0<\cdots<\lambda_T=1\), set
\[
\pi_t(\theta)\propto
q_0(\theta)^{1-\lambda_t}\pi(\theta)^{\lambda_t}.
\]
Fix a particle count \(N_p\) and an ESS resampling threshold \(\tau_{\mathrm{ESS}}\in(0,N_p]\); initialize independent \(\theta_0^i\sim q_0\) with \(w_0^i=1/N_p\).
For \(t=1,\ldots,T\):
\begin{enumerate}
\item Reweight and normalize:
\(\widetilde w_t^i\propto w_{t-1}^i\left[\pi(\theta_{t-1}^i)/q_0(\theta_{t-1}^i)\right]^{\lambda_t-\lambda_{t-1}}\).
\item If \((\sum_i(\widetilde w_t^i)^2)^{-1}<\tau_{\mathrm{ESS}}\), draw ancestor indices \(A_t^i\) with probabilities \(\widetilde w_t^j\), and set \((\bar\theta_t^i,\bar w_t^i)=(\theta_{t-1}^{A_t^i},1/N_p)\). Otherwise set \((\bar\theta_t^i,\bar w_t^i)=(\theta_{t-1}^i,\widetilde w_t^i)\). \item Apply a \(\pi_t\)-invariant mutation kernel \(K_t\), possibly comprising several MCMC steps: draw \(\theta_t^i\sim K_t(\bar\theta_t^i,\cdot)\) and retain \(w_t^i=\bar w_t^i\).
\end{enumerate}
Return the weighted posterior particles
\(\{\theta_T^i,w_T^i\}_{i=1}^{N_p}\).
\end{algorithm}
SMC performance depends on initialization, tempering, resampling, and mutation. Large tempering steps can concentrate weight on a small number of particles, while many small steps increase computation. Resampling can reduce representation of weakly populated regions, but local mutation may leave separated regions or long flat directions poorly explored. Simulation~\ref{sim:conditional_gaussian}
uses prior initialization for \(k=2\) and localized Gaussian initialization for \(k=10\). The latter case compares SMC with direct sampling along the equivalence fiber.

\section{Additional Theoretical Results}
\label{append:additional_results}

This section supplements the comparisons in Section~\ref{subsec: convergence}. It gives further kernel comparisons, calculations for the illustrative examples, extensions for separated regions, and a separate comparison for expanding product fibers. The proofs use the supporting results in Appendix~\ref{append:supporting_lemmas}.

\subsection{Kernel comparisons}
\label{oa:kernel_comparisons}
\begin{proposition}[Reversibility of the random-order kernel]
\label{prop:reversibility}
Let \(P\) and \(T\) be \(\pi\)-reversible Markov kernels, and define
\[
\bar P=\tfrac12(PT)+\tfrac12(TP).
\]
Then \(\bar P\) is \(\pi\)-reversible. In particular, the conclusion applies when \(T\) is the kernel in Lemma~\ref{lemma:conditional_sampling}.
\end{proposition}

\begin{proof}
For bounded measurable \(f,g\), reversibility of \(P\) and \(T\) gives
\[
\langle f,PTg\rangle_\pi
=
\langle Pf,Tg\rangle_\pi
=
\langle TPf,g\rangle_\pi.
\]
Thus the adjoint of \(PT\) is \(TP\), so \(\frac12(PT+TP)\) is self-adjoint and hence reversible.
\end{proof}

\begin{proposition}[Spectral comparison for the hybrid kernel]
\label{prop:IA_master_bound}
Let \(P\) and \(T\) be \(\pi\)-reversible Markov kernels and define
\[
\bar P=\tfrac12(PT+TP),
\qquad
\widetilde P_H=(1-\varepsilon)P+\varepsilon T.
\]
For the Poincar\'e spectral gap,
\[
\gamma(\widetilde P_H)\ge \varepsilon(1-\varepsilon)\gamma(\bar P).
\]
In particular, for \(\varepsilon=1/2\),
\[
\gamma\!\left(\tfrac12P+\tfrac12T\right)\ge \tfrac14\gamma(\bar P).
\]
\end{proposition}

\begin{proof}
For unit \(f\in L_2^0(\pi)\), set
\(a=1-\langle Pf,f\rangle_\pi\) and
\(b=1-\langle Tf,f\rangle_\pi\).
Reversible contractions satisfy
\(\|(I-P)f\|_2^2\le2a\) and
\(\|(I-T)f\|_2^2\le2b\). Consequently,
\[
\begin{aligned}
1-\langle\bar P f,f\rangle_\pi
&=a+b-\langle(I-T)f,(I-P)f\rangle_\pi\\
&\le a+b+2\sqrt{ab}=(\sqrt a+\sqrt b)^2.
\end{aligned}
\]
Moreover,
\[
(1-\varepsilon)a+\varepsilon b
-\varepsilon(1-\varepsilon)(\sqrt a+\sqrt b)^2
=\bigl((1-\varepsilon)\sqrt a-\varepsilon\sqrt b\bigr)^2
\ge0.
\]
Thus
\(1-\langle\widetilde P_Hf,f\rangle_\pi
\ge\varepsilon(1-\varepsilon)
[1-\langle\bar Pf,f\rangle_\pi]\).
Taking the infimum over unit \(f\in L_2^0(\pi)\) proves the claim.
\end{proof}
 
\subsection{Finite-state and two-mode examples}
\label{oa:modal_examples}
\begin{proposition}[Two parameters, finite states]
\label{prop:gibbs_gap}
Fix an integer \(m\ge2\), let \(\Omega=\{1,\ldots,m\}^2\), and define
\[
\pi(\theta_1,\theta_2)=
\begin{cases}
a,&\theta_1=\theta_2,\\
b,&\theta_1\ne\theta_2,
\end{cases}
\qquad 0<b<a,\qquad ma+m(m-1)b=1.
\]
Let \(D=\{(i,i):i=1,\ldots,m\}\), \(O=\Omega\setminus D\), and let \(P\) be the random-scan single-site Gibbs kernel, choosing each coordinate with probability \(1/2\). The teleport kernel \(T\) draws uniformly from \(D\) when the current state belongs to \(D\), and leaves every state in \(O\) unchanged. Then \(T\) is \(\pi\)-reversible and, for fixed
\(\varepsilon\in(0,1)\), the hybrid
\(\widetilde P_H=(1-\varepsilon)P+\varepsilon T\) satisfies
\[
\gamma(P)\longrightarrow0,
\qquad
\gamma(\widetilde P_H)
\longrightarrow\min\{\varepsilon,1-\varepsilon\}
\qquad\text{as }b\to0.
\]
\end{proposition}

\begin{proof}
Detailed balance for \(T\) follows from the equal probabilities of the diagonal states. For \(S=\{(1,1)\}\), we have
\(\pi(S)=a<1/m\le1/2\), and every exit from \(S\) enters \(O\).
Its conductance is therefore
\(P((1,1),O)=(m-1)b/[a+(m-1)b]\to0\).
Cheeger's upper bound gives \(\gamma(P)\to0\).

At \(b=0\), the limiting Gibbs kernel fixes each diagonal state and moves an off-diagonal state \((i,j)\) to \((i,i)\) or \((j,j)\), each with probability \(1/2\).
Ordering the full \(m^2\)-state space as \(D,O\), the limiting hybrid matrix is block lower triangular with diagonal blocks
\((1-\varepsilon)I_m+\varepsilon U_m\) and
\(\varepsilon I_{m(m-1)}\), where
\(U_m=m^{-1}\mathbf1_m\mathbf1_m'\).
Its eigenvalues are \(1\), \(1-\varepsilon\) with multiplicity \(m-1\), and \(\varepsilon\) with multiplicity \(m(m-1)\).
For \(b>0\), the hybrid matrix is reversible and its entries converge to this limiting matrix. Continuity of the eigenvalues therefore gives
\(\gamma(\widetilde P_H)\to
1-\max\{1-\varepsilon,\varepsilon\}
=\min\{\varepsilon,1-\varepsilon\}\).
\end{proof}

\begin{lemma}[One-dimensional \(\delta\)-ball random walk on a Laplace interval]
    \label{lemma: lb_deltaball}
Let \(L,\nu,\delta>0\), and let \(P_{L,\nu,\delta}\) be the Metropolis--Hastings kernel on \([-L,L]\) with target density
\[
\pi_{L,\nu}(x)\propto e^{-\nu|x|},
\qquad x\in[-L,L],
\]
and symmetric \(\delta\)-ball proposal, with proposals leaving \([-L,L]\) rejected. There exist universal constants \(C>0\) and \(c_*>0\) such that if
\[
0<\delta\le L,
\qquad
\nu\delta\le c_*,
\]
then
\[
\gamma(P_{L,\nu,\delta})\ge C(\nu\delta)^2,
\]
where \(C\) and \(c_*\) do not depend on \(L\).
\end{lemma}

\begin{proof}
Set \(h=\nu\delta\) and \(a=\nu L\). After the change of variable \(u=\nu x\), the target has density \(p_a(u)\propto e^{-|u|}\) on \([-a,a]\), the proposal radius is \(h\), and \(a\ge h\). It is enough to consider \(h\le1\), so take \(c_*=1\).

We first record a uniform expansion bound for the truncated-Laplace law \(\pi_a\). If \(F_a\) is its distribution function, then
\[
 p_a(F_a^{-1}(t))
 =
 \min\{t,1-t\}
 +
 \frac{e^{-a}}{2(1-e^{-a})}.
\]
This function is symmetric and subadditive. Therefore, the one-dimensional isoperimetric inequality of \citet{bobkov1998isoperimetric} implies that, if \(\pi_a(B)=b\le1/2\) and \(B^r\) is the open \(r\)-neighborhood of \(B\),
\[
\pi_a(B^r)
\ge
F_a(F_a^{-1}(b)+r).
\]
Since \(\log p_a\) is one-Lipschitz, for \(0<r\le\min\{1/2,a\}\),
\begin{equation}
\label{eq:truncated_laplace_expansion}
\pi_a(B^r)-\pi_a(B)
\ge
r e^{-r}p_a(F_a^{-1}(b))
\ge
e^{-1/2}rb.
\end{equation}

The off-diagonal density of the Metropolis kernel from \(u\) to \(z\) is
\[
k_u(z)=\frac{1}{2h}
\min\left\{1,\frac{p_a(z)}{p_a(u)}\right\}
\mathbf 1\{|z-u|\le h,\ z\in[-a,a]\}.
\]
On its support, \(k_u(z)\ge e^{-1}/(2h)\). If \(|u-v|<h/2\), the supports of \(k_u\) and \(k_v\) overlap on an interval of length at least \(h\), because \(a\ge h\). Their common off-diagonal mass is therefore at least \(e^{-1}/2\).

Fix a measurable cut \(A\), write \(Q(A,A^c)=\int_A P(u,A^c)\pi_a(du)\), and set \(\eta=e^{-1}/8\). Define
\[
A_1=\{u\in A:P(u,A^c)<\eta\},
\qquad
A_2=\{u\in A^c:P(u,A)<\eta\}.
\]
The overlap bound implies \(\operatorname{dist}(A_1,A_2)\ge h/2\). Otherwise the common mass for two such points \(u\in A_1\) and \(v\in A_2\) would be at most
\[
P(u,A^c)+P(v,A)<2\eta=e^{-1}/4,
\]
a contradiction. Set
\[
R=[-a,a]\setminus(A_1\cup A_2),
\qquad
s=\min\{\pi_a(A),\pi_a(A^c)\}.
\]
If either \(\pi_a(A\setminus A_1)\) or \(\pi_a(A^c\setminus A_2)\) is at least \(s/2\), reversibility gives \(Q(A,A^c)\ge\eta s/2\). Otherwise \(\pi_a(A_1)\wedge\pi_a(A_2)\ge s/2\). Apply \eqref{eq:truncated_laplace_expansion}, with \(r=h/2\), to the smaller of \(A_1\) and \(A_2\). Its \(r\)-neighborhood cannot meet the other set, so
\[
\pi_a(R)
\ge
\frac{e^{-1/2}h}{2}
\min\{\pi_a(A_1),\pi_a(A_2)\}
\ge
\frac{e^{-1/2}h}{4}s.
\]
Moreover, reversibility gives \(2Q(A,A^c)\ge\eta\pi_a(R)\). Combining the two cases yields
\[
Q(A,A^c)
\ge
c_1h\min\{\pi_a(A),\pi_a(A^c)\}
\]
for a universal \(c_1>0\). Hence the conductance is at least \(c_1h\), and Cheeger's inequality gives
\[
\gamma(P_{L,\nu,\delta})
\ge
\frac12c_1^2h^2
=
C(\nu\delta)^2.
\]
All constants are independent of \(L\).
\end{proof}

\begin{proposition}[One-dimensional two-mode target]
\label{prop:1ddeltaball}
Let the state space be a one-dimensional circle \(\Omega\) of circumference \(4L\), represented by \([-2L,2L]\) with endpoints connected. The target \(\pi\) is bimodal:
\[
\pi(\theta)\propto
\begin{cases}
e^{-\nu|\theta|}, & \theta\in[-L,L],\\
e^{-\nu(2L-|\theta|)}, & \theta\in[-2L,-L)\cup(L,2L],
\end{cases}
\]
with \(\nu>0\). Let \(P\) be Metropolis--Hastings with a symmetric \(\delta\)-ball random-walk proposal.

\begin{enumerate}
\item \textit{Standard sampler.} There is \(C_1>0\) such that
\[
\gamma(P)\le C_1 e^{-\nu(L-\delta)}.
\]

\item \textit{Hybrid identification-aware RWM.} Define
\[
s(\theta)=
\begin{cases}
\theta+2L, & \theta\in[-2L,0),\\[2pt]
\theta-2L, & \theta\in[0,2L),
\end{cases}
\]
and
\[
T(\theta,\cdot)
=
\frac12\delta_\theta(\cdot)+\frac12\delta_{s(\theta)}(\cdot).
\]
Let
\[
\widetilde P_H=(1-\varepsilon)P+\varepsilon T,
\qquad \varepsilon\in(0,1).
\]
There exist constants \(C_2>0\) and \(c_0>0\), independent of \(L\), such that whenever \(0<\delta\le L\) and \(\nu\delta\le c_0\),
\[
\gamma(\widetilde P_H)
\ge
C_2\,\varepsilon(1-\varepsilon)(\nu\delta)^2.
\]
In particular, choosing \(\delta\asymp\nu^{-1}\) yields a positive lower bound independent of the mode separation \(L\).
\end{enumerate}
\end{proposition}

\begin{proof}
\textit{Part (1): local RWM.}
Write \(A=[-L,L]\) and \(A^c=\Omega\setminus A\). By symmetry, \(\pi(A)=\pi(A^c)=1/2\). Define the inner boundary strip
\[
\partial_\delta A
:=
\{\theta\in A:\operatorname{dist}(\theta,A^c)\le\delta\}.
\]
A move from \(A\) to \(A^c\) can occur only from
\[
\partial_\delta A
=
A\cap\bigl([L-\delta,L]\cup[-L,-L+\delta]\bigr).
\]
Therefore
\[
\begin{aligned}
\Phi_P(A)
&:=
\frac{1}{\pi(A)}
\int_A P(\theta,A^c)\pi(d\theta) \\
&\le
\frac{1}{\pi(A)}
\int_{\partial_\delta A}\pi(\theta)d\theta
\le
C e^{-\nu(L-\delta)}.
\end{aligned}
\]
Since \(\mathbf h_P\le\Phi_P(A)\), Cheeger's inequality gives
\[
\gamma(P)\le2\mathbf h_P\le2\Phi_P(A)\le C_1e^{-\nu(L-\delta)}.
\]

\medskip
\textit{Part (2): hybrid identification-aware RWM.}
Use the partition \(\{A,A^c\}\). For \(\pi\)-almost every \(\theta\in A\),
\[
T(\theta,A^c)=\frac12,
\]
and for \(\pi\)-almost every \(\theta\in A^c\),
\[
T(\theta,A)=\frac12.
\]
Hence the two-state chain on \(\{A,A^c\}\) induced by \(\widetilde P_H\) has both off-diagonal transition probabilities at least \(\varepsilon/2\). Its spectral gap is therefore at least \(\varepsilon\).

It remains to control the restricted chains inside \(A\) and \(A^c\). Let \(P_A\) denote the \(\delta\)-ball MH kernel restricted to \(A\), targeting
\[
\pi_A(\theta)\propto e^{-\nu|\theta|},\qquad \theta\in[-L,L],
\]
and define \(P_{A^c}\) analogously. By Lemma~\ref{lemma: lb_deltaball}, there exist constants \(C'>0\) and \(c_0>0\) such that, whenever \(0<\delta\le L\) and \(\nu\delta\le c_0\),
\[
\gamma(P_A)\ge C'(\nu\delta)^2,
\qquad
\gamma(P_{A^c})\ge C'(\nu\delta)^2.
\]
The second bound follows from the first by shifting the circular coordinate by \(2L\), which identifies \(A^c\) with another copy of \([-L,L]\) carrying the same density shape.
The restricted hybrid chain dominates the local-kernel contribution:

\[
(\widetilde P_H)_A\ge (1-\varepsilon)P_A,
\qquad
(\widetilde P_H)_{A^c}\ge (1-\varepsilon)P_{A^c}
\]
in Dirichlet form. Therefore
\[
\gamma((\widetilde P_H)_A)\ge (1-\varepsilon)C'(\nu\delta)^2,
\qquad
\gamma((\widetilde P_H)_{A^c})\ge (1-\varepsilon)C'(\nu\delta)^2.
\]
Applying the state-decomposition theorem,
\[
\gamma(\widetilde P_H)
\ge
\frac14\cdot \varepsilon\cdot (1-\varepsilon)C'(\nu\delta)^2.
\]
This gives the stated lower bound with \(C_2=C'/4\).
\end{proof}

\subsection{Further results for separated regions}
\label{oa:separated_extensions}
\begin{corollary}[Teleportation over random subsets of modes]
\label{cor:k_mode_random_subset} 
Let \(A_1,\ldots,A_m\), \(m\ge2\), be the modal partition in Proposition~\ref{prop:deltaMH}, and let \(\chi\) be as in \eqref{eqn:teleport_kernel}. Assume \(\pi(A_i\mid\chi=u)=1/m\) for every \(i\) and almost every \(u\) under the marginal law of \(\chi\), with sampling from
\(\pi(\cdot\mid\chi=u,A_i)\) available. Let \(P\) be a \(\pi\)-reversible local kernel satisfying \(\min_i\gamma(P_{A_i})\ge c_{\mathrm{within}}>0\),
with \(c_{\mathrm{within}}\) independent of \(d_*\).

Fix \(k\in\{2,\ldots,m\}\), and write \([m]=\{1,\ldots,m\}\). From \(\theta\in A_i\), choose a uniform \(k\)-element subset \(S\subseteq[m]\) containing \(i\), choose \(J\) uniformly from \(S\), and draw \(\theta'\sim\pi(\cdot\mid\chi=\chi(\theta),A_J)\). The resulting teleport kernel \(T_k\) is \(\pi\)-reversible and preserves observational equivalence.
For \(\varepsilon\in(0,1)\), let \(\widetilde P_{k,H}=(1-\varepsilon)P+\varepsilon T_k\). Then
\[
\gamma(\widetilde P_{k,H})
\ge
\frac14\,\varepsilon(1-\varepsilon)
\frac{m(k-1)}{k(m-1)}\,c_{\mathrm{within}},
\]
independently of \(d_*\).
\end{corollary}

\begin{proof}[Proof of Corollary~\ref{cor:k_mode_random_subset}]
The probability of selecting \(J=i\) is \(1/k\). For \(j\ne i\), the probability that \(j\in S\) is \((k-1)/(m-1)\), followed by selection with probability \(1/k\). Thus the symmetric component matrix \(T_k^{\mathrm{comp}}\) satisfies
\[
T_k^{\mathrm{comp}}(i,i)=\frac1k,\qquad
T_k^{\mathrm{comp}}(i,j)=\frac{k-1}{k(m-1)}
\quad(j\ne i).
\]

For almost every \(u\), write
\(\pi_u=\pi(\cdot\mid\chi=u)\) and
\(\pi_{u,i}=\pi(\cdot\mid\chi=u,A_i)\).
On \(A_i\times A_j\),
\[
\pi_u(d\theta)T_k(\theta,d\theta')
=
\frac{T_k^{\mathrm{comp}}(i,j)}{m}
\pi_{u,i}(d\theta)\pi_{u,j}(d\theta').
\]
Symmetry gives detailed balance conditional on \(\chi=u\). Integrating over the marginal law of \(\chi\) establishes \(\pi\)-reversibility of \(T_k\), and hence of \(\widetilde P_{k,H}\). Since \(\chi(\theta')=\chi(\theta)\) almost surely, the update preserves observational equivalence.

The conditional mass assumption implies \(\pi(A_i)=1/m\). Thus the component chain
\[
R(i,j)=m\int_{A_i}P(\theta,A_j)\pi(d\theta)
\]
is reversible under the uniform law on \([m]\). Set
\[
g_k=\frac{m(k-1)}{k(m-1)},\qquad
U_m=\frac1m\mathbf1\mathbf1^\top.
\]
Since \(T_k^{\mathrm{comp}}=(1-g_k)I+g_kU_m\), the hybrid component chain is
\[
Q_{\varepsilon,k}
=(1-\varepsilon)R+\varepsilon T_k^{\mathrm{comp}}.
\]
For every zero-mean, unit-norm \(f\) under the uniform law,
\[
\langle Q_{\varepsilon,k}f,f\rangle_{\mathrm{unif}}
=
(1-\varepsilon)\langle Rf,f\rangle_{\mathrm{unif}}
+\varepsilon(1-g_k)
\le 1-\varepsilon g_k.
\]
Hence \(\gamma(Q_{\varepsilon,k})\ge\varepsilon g_k\).

Restriction is linear:
\[
(\widetilde P_{k,H})_{A_i}
=(1-\varepsilon)P_{A_i}+\varepsilon(T_k)_{A_i}.
\]
By nonnegativity of the second kernel's Dirichlet form,
\[
\gamma\!\left((\widetilde P_{k,H})_{A_i}\right)
\ge (1-\varepsilon)\gamma(P_{A_i})
\ge (1-\varepsilon)c_{\mathrm{within}}.
\]
Lemma~\ref{lemma: state_decomp} therefore gives
\[
\gamma(\widetilde P_{k,H})
\ge
\frac14\,\gamma(Q_{\varepsilon,k})
\min_i\gamma\!\left((\widetilde P_{k,H})_{A_i}\right)
\ge
\frac14\,\varepsilon(1-\varepsilon)g_kc_{\mathrm{within}},
\]
as claimed.
\end{proof}

\begin{proposition}[Gaussian RWM]
\label{prop:gaussian_rw}
Let \(\pi\) be a probability density on \(\Omega\subset\mathbb R^n\). Suppose \(A_1,\ldots,A_m\), with \(m\ge2\), form a measurable partition of \(\Omega\) into modal regions with \(\pi(A_i)>0\). For each \(i\), set
\[
w_i:=\pi(A_i),
\qquad
\pi_i(B):=\frac{\pi(B\cap A_i)}{w_i},
\qquad
A_i^c:=\Omega\setminus A_i.
\]
Let \(\mu_i\in A_i\) be a mode of \(\pi\) restricted to \(A_i\), and assume
\[
d_i:=\operatorname{dist}(\mu_i,A_i^c)>0,
\qquad
d_*:=\min_i d_i.
\]

Assume:
\begin{enumerate}
\item \label{prop:gaussian_rw1} For each \(i\), there exist constants \(c_i\ge1\) and \(\nu_i>0\), not depending on \(d_*\), such that for all \(r\in(0,d_i]\),
\[
\pi_i\left(\{\theta\in A_i:\|\theta-\mu_i\|\ge r\}\right)
\le
c_i e^{-\nu_i r}.
\]
Write
\[
c_{\max}:=\max_i c_i,
\qquad
\nu_{\min}:=\min_i\nu_i.
\]

\item \label{prop:gaussian_rw2} The teleport kernel \(T\) is \(\pi\)-reversible. Its induced component kernel
\[
T_H(i,j)
:=
\frac{1}{w_i}\int_{A_i}T(\theta,A_j)\pi(d\theta),
\qquad i\ne j,
\]
satisfies
\[
\tau_T:=\min_{i\ne j}T_H(i,j)>0,
\]
with \(\tau_T\) not depending on \(d_*\).

\item \label{prop:gaussian_rw3} Let \(P\) be the Metropolis--Hastings kernel with Gaussian random-walk proposal
\[
q(\theta,\cdot)=\mathcal N(\theta,\sigma^2I_n),
\]
and let \(P_{A_i}\) denote its restriction to \(A_i\). There exists \(c_{\mathrm{loc}}>0\), not depending on \(d_*\), such that
\[
\min_{1\le i\le m}\gamma(P_{A_i})\ge c_{\mathrm{loc}}.
\]
\end{enumerate}
Together with \(\pi\)-reversibility of \(T\), Assumption~\ref{assump:multimodality}\ref{assump:multimodality2} implies condition~\ref{prop:gaussian_rw2} with \(\tau_T\ge\varepsilon_1\varepsilon_2\).

Define
\[
\widetilde P_H=(1-\varepsilon)P+\varepsilon T,
\qquad \varepsilon\in(0,1).
\]
Then:
\begin{enumerate}
\item For any fixed \(\sigma>0\), there exist constants \(C<\infty\) and \(d_0<\infty\), not depending on \(d_*\), such that for all \(d_*\ge d_0\),
\[
\gamma(P)
\le
C\exp\left(-\frac12\nu_{\min}d_*\right).
\]

\item The hybrid kernel satisfies
\[
\gamma(\widetilde P_H)
\ge
\frac18\,\varepsilon^2\tau_T^2(1-\varepsilon)c_{\mathrm{loc}}.
\]
In particular, \(\gamma(\widetilde P_H)\) is bounded below by a positive constant independent of \(d_*\).
\end{enumerate}
\end{proposition}

\begin{remark}
Strong log-concavity can help verify condition~\ref{prop:gaussian_rw1}. Condition~\ref{prop:gaussian_rw3} separately requires a positive lower bound on the spectral gap within each modal region. The proposition therefore separates movement between regions from mixing within each region.
\end{remark}

\begin{proof}
Choose \(i\) with \(w_i\le1/2\), and split \(A_i\) at radius
\(d_i/2\) around \(\mu_i\).
The outer part has conditional mass at most
\(c_i e^{-\nu_i d_i/2}\). From the inner part, an exit requires
a proposal jump of length at least \(d_i/2\).
For \(Z\sim\mathcal N(0,\sigma^2I_n)\), a Chernoff bound gives
\[
\frac{\int_{A_i}P(\theta,A_i^c)\pi(d\theta)}{w_i}
\le 2^{n/2}e^{-d_i^2/(16\sigma^2)}
     +c_i e^{-\nu_i d_i/2}.
\]
For \(d_*\ge8\sigma^2\nu_{\min}\), Cheeger's upper bound yields
\[
\gamma(P)\le
2(2^{n/2}+c_{\max})e^{-\nu_{\min}d_*/2}.
\]

Let \(R\) be the component chain induced by \(\widetilde P_H\).
Its off-diagonal probabilities satisfy
\(R(i,j)\ge\varepsilon T_H(i,j)\ge\varepsilon\tau_T\).
Thus \(\gamma(R)\ge\varepsilon^2\tau_T^2/2\).
Restriction is linear in the kernel, so comparison of Dirichlet
forms gives
\(\gamma((\widetilde P_H)_{A_i})
 \ge(1-\varepsilon)c_{\mathrm{loc}}\).
Lemma~\ref{lemma: state_decomp} therefore yields
\[
\gamma(\widetilde P_H)
\ge\frac18\varepsilon^2\tau_T^2
             (1-\varepsilon)c_{\mathrm{loc}}.
\]
\end{proof}

The next result gives the corresponding separation comparison for a fixed-length Hamiltonian kernel, followed by a conditional extension to fixed-tuning NUTS.

\begin{assumption}
\label{assump:hmc_displacement}
Let \(P_{\mathrm{HMC}}\) be a \(\pi\)-reversible, fixed-length,
Metropolis-corrected HMC kernel with fixed step size, mass matrix, and number
of leapfrog steps. Let \(\theta'\) denote its raw position proposal before the Metropolis accept/reject step. There exist constants \(C_H,c_H,\tau_H>0\), independent of the inter-mode separation \(d_*\), such that for every region \(A_i\), every \(R\in(0,d_i)\), every
\[
\theta\in\{x\in A_i:\|x-\mu_i\|\le R\},
\]
and every \(a>0\),
\[
\Pr\!\left(\|\theta'-\theta\|\ge a\mid\theta\right)
\le
C_H\exp\!\left(-\frac{c_Ha^2}{\tau_H^2}\right).
\]
The scale \(\tau_H\) may depend on the fixed HMC tuning and target constants, but not on \(d_*\).
\end{assumption}

\begin{proposition}[Fixed-length HMC in separated regions]
\label{prop:HMC}
Under the target and modal-partition setup of Proposition~\ref{prop:deltaMH}, suppose Assumption~\ref{assump:multimodality}\ref{assump:multimodality1} and Assumption~\ref{assump:hmc_displacement} hold. Then there are constants \(C,c>0\), independent of \(d_*\), such that
\[
\gamma(P_{\mathrm{HMC}})
\le
C\exp\!\left[
-\min\left\{
\frac12\nu_{\min}d_*,
\frac{cd_*^2}{\tau_H^2}
\right\}
\right].
\]
If, in addition, Assumption~\ref{assump:multimodality}\ref{assump:multimodality2} holds for \(T\), Assumption~\ref{assump:multimodality}\ref{assump:multimodality3} holds for \(P=P_{\mathrm{HMC}}\), and \(T\) is \(\pi\)-reversible, then, for any fixed \(\varepsilon\in(0,1)\), the hybrid
\[
\widetilde P_{\mathrm{HMC,IA}}
=
(1-\varepsilon)P_{\mathrm{HMC}}+\varepsilon T
\]
satisfies
\[
\gamma(\widetilde P_{\mathrm{HMC,IA}})\ge c_0
\]
for some \(c_0>0\) independent of \(d_*\).
\end{proposition}

\begin{proof}
Fix a region \(A_i\) with \(\pi(A_i)\le1/2\), and write \(A_i^c=\Omega\setminus A_i\). For \(R\in(0,d_i)\), define
\[
C_R:=\{\theta\in A_i:\|\theta-\mu_i\|<R\},
\qquad
S_R:=A_i\setminus C_R.
\]
Then
\[
\int_{A_i}P_{\mathrm{HMC}}(\theta,A_i^c)\pi(d\theta)
=
I_{\mathrm{core}}+I_{\mathrm{shell}},
\]
where the two terms integrate over \(C_R\) and \(S_R\), respectively. Assumption~\ref{assump:multimodality}\ref{assump:multimodality1} gives
\[
I_{\mathrm{shell}}
\le
\pi(S_R)
\le
\pi(A_i)c_{\max}e^{-\nu_{\min}R}.
\]
For \(\theta\in C_R\), any proposal reaching \(A_i^c\) must move by at least \(d_i-R\). Assumption~\ref{assump:hmc_displacement} therefore gives
\[
I_{\mathrm{core}}
\le
\pi(A_i)C_H
\exp\!\left[-\frac{c_H(d_i-R)^2}{\tau_H^2}\right].
\]
The Metropolis step cannot increase the probability of an accepted move into \(A_i^c\). Setting \(R=d_i/2\), using \(d_i\ge d_*\), and applying the upper half of Cheeger's inequality yields the stated bound after absorbing fixed constants into \(C\) and \(c\).

For the hybrid result, Assumption~\ref{assump:multimodality}\ref{assump:multimodality2} gives a component-chain gap bounded below independently of \(d_*\). Assumption~\ref{assump:multimodality}\ref{assump:multimodality3}, applied with \(P=P_{\mathrm{HMC}}\), gives the uniform restricted-chain minorization. The same state-decomposition argument used in the proof of Proposition~\ref{prop:deltaMH} then gives \(\gamma(\widetilde P_{\mathrm{HMC,IA}})\ge c_0\).
\end{proof}

\begin{corollary}[Fixed-tuning NUTS]
\label{cor:fixed_nuts}
Let \(P_N\) be a \(\pi\)-reversible NUTS kernel with fixed step size, fixed mass matrix, and fixed maximum tree depth. Suppose the state \(\theta'\) selected for the next transition satisfies, for some \(C_N,c_N,\tau_N>0\) independent of \(d_*\),
\[
\Pr\!\left(\|\theta'-\theta\|\ge a\mid\theta\right)
\le
C_N\exp\!\left(-\frac{c_Na^2}{\tau_N^2}\right)
\]
for every \(A_i\), \(R\in(0,d_i)\), \(\theta\in\{x\in A_i:\|x-\mu_i\|\le R\}\), and \(a>0\). Under Assumption~\ref{assump:multimodality}\ref{assump:multimodality1}, there are constants \(C,c>0\), independent of \(d_*\), such that
\[
\gamma(P_N)
\le
C\exp\!\left[
-\min\left\{
\frac12\nu_{\min}d_*,
\frac{cd_*^2}{\tau_N^2}
\right\}
\right].
\]
If Assumption~\ref{assump:multimodality}\ref{assump:multimodality2} also holds for \(T\), Assumption~\ref{assump:multimodality}\ref{assump:multimodality3} holds for \(P=P_N\), and \(T\) is \(\pi\)-reversible, then, for any fixed \(\varepsilon\in(0,1)\), the hybrid
\[
\widetilde P_{N,\mathrm{IA}}
=
(1-\varepsilon)P_N+\varepsilon T
\]
satisfies
\[
\gamma(\widetilde P_{N,\mathrm{IA}})\ge c_0
\]
for some \(c_0>0\) independent of \(d_*\).
\end{corollary}

\begin{proof}
Apply the core-and-shell conductance argument from Proposition~\ref{prop:HMC} to the selected next-state displacement bound. The hybrid conclusion follows from the same component-chain and restricted-chain decomposition.
\end{proof}

The NUTS corollary requires the displayed displacement bound for the state selected by the complete transition. A bound on tree depth alone does not imply that condition.

\subsection{Expanding product fibers}
\label{oa:product_fibers}
We finish with the expanding-product comparison used in Example~\ref{example:flatY}.

\begin{proposition}[Expanding product fibers]
\label{prop:local_unid_flatY_bar}
Let \(m,r\in\mathbb N\) and \(D>0\). Define
\[
\Theta_D = X \times Y_D \subset \mathbb{R}^{m+r},
\]
where \(X\subset\mathbb{R}^m\) is bounded and convex with nonempty interior, and
\[
Y_D = D\cdot G_0,
\qquad
G_0\subset\mathbb{R}^r.
\]
Assume \(G_0\) is bounded with \(0<|G_0|<\infty\) and that there exists \(C_{\partial}<\infty\) such that
\[
|\{y\in G_0:\operatorname{dist}(y,\partial G_0)\le \eta\}|
\le C_{\partial}\eta,
\qquad \eta\in(0,1].
\]
Assume also that there exist \(t\in\mathbb R\) and \(C_{\mathrm{slab}}<\infty\) such that
\[
|\{y\in G_0:\ y_1\le t\}|=\tfrac12|G_0|,
\qquad
|\{y\in G_0:\ |y_1-t|\le s\}|\le C_{\mathrm{slab}}s,
\quad s\in(0,1].
\]
Then \(|Y_D|=|G_0| D^r\).

Let \(U_D\) be the uniform probability law on \(Y_D\), and let
\[
u_D(y):=\frac{\mathbf 1_{Y_D}(y)}{|Y_D|}
\]
be its density with respect to Lebesgue measure. Let the target density factor as
\[
\pi_D(\theta_x,\theta_y)=p(\theta_x)u_D(\theta_y),
\]
where \(p\) is a continuous probability density on \(X\) satisfying
\[
0<p_{\min}\le p(\theta_x)\le p_{\max}<\infty
\qquad \forall \theta_x\in X.
\]
Let the local RWM proposal density be
\[
q(\theta,\theta') = g(\theta'-\theta),
\]
where \(g\) is a symmetric probability density. Assume there exist \(\delta>0\) and \(c_g>0\) such that \(g(z)\ge c_g\) for all \(\|z\|\le\delta\). Also define
\[
\overline G(t):=\int_{\|z\|>t}g(z)\,dz
\]
and assume
\[
\int_0^{\infty}\overline G(t)\,dt\le T_g<\infty.
\]
These conditions hold for Gaussian proposals, compactly supported proposals, and many proposals with exponentially decaying tails.
For all sufficiently large \(D\), let
\[
Y_{\mathrm{int}}:=\{y\in Y_D:\operatorname{dist}(y,\partial Y_D)\ge\delta/2\}.
\]
Assume there exist \(n_X\in\mathbb N\), \(\eta_X>0\), a compact set \(R_x\) contained in \(\operatorname{int}(X)\), and a probability measure \(\nu_X\) supported on \(R_x\), all independent of \(D\), such that, for every \((x,y)\in X\times Y_{\mathrm{int}}\),
\[
P^{n_X}((x,y),B\times Y_D)\ge \eta_X\nu_X(B)
\]
for all measurable \(B\subseteq R_x\), where \(P\), defined below, is the local RWM kernel.

Consider two Markov kernels. First, \(P\) denotes the local RWM on \(\Theta_D\) using proposal \(g\) and standard Metropolis--Hastings acceptance. Second, define the teleport kernel
\[
T\big((\theta_x,\theta_y),A\big)
=
\int \mathbf 1_A(\theta_x,\widetilde y)\,U_D(d\widetilde y).
\]
Set
\[
\widetilde P_{H,D}:=(1-\varepsilon)P+\varepsilon T,\qquad \varepsilon\in(0,1).
\]

Then:
\begin{enumerate}
\item There exists a constant \(C_1<\infty\), independent of \(D\), such that the spectral gap of \(P\) satisfies
\[
\gamma(P)\ \le\ \frac{C_1}{D}.
\]
In particular, \(\gamma(P)=O(D^{-1})\), and hence \(\gamma(P)\to0\) as \(D\to\infty\).
\item There exist constants \(c_0>0\) and \(D_0<\infty\), independent of \(D\), such that for all \(D\ge D_0\),
\[
\gamma(\widetilde P_{H,D})\ge c_0.
\]
\end{enumerate}
\end{proposition}

\begin{proof}
\textit{(1) Upper bound for \(P\).}
Let \(t\) be the volume-bisecting cut in the slab condition and set \(A=X\times\{y\in Y_D:y_1\le Dt\}\). Since \(Y_D=DG_0\), scaling preserves the fraction of volume on either side of the cut. Together with \(\int_Xp(x)\,dx=1\), this gives \(\pi_D(A)=1/2\).

For \(Y\sim U_D\), let \(S(Y)=|Y_1-Dt|\) be its distance from the cut. Using \(|Y_D|=D^r|G_0|\), for \(0\le s\le D\) we obtain
\[
\begin{aligned}
\Pr\{S(Y)\le s\}
&=\frac{|\{y\in Y_D:|y_1-Dt|\le s\}|}{|Y_D|}\\
&=\frac{D^r|\{z\in G_0:|z_1-t|\le s/D\}|}{D^r|G_0|}\\
&\le\frac{C_{\mathrm{slab}}}{|G_0|}\frac{s}{D}.
\end{aligned}
\]
For \(s>D\), use \(\Pr\{S(Y)\le s\}\le1\). Thus, setting \(C_0=\max\{1,C_{\mathrm{slab}}/|G_0|\}\), we have
\[
\Pr\{S(Y)\le s\}\le C_0\min\{1,s/D\},
\qquad s\ge0.
\]
Let \(Z\sim g\) be independent of \(Y\), and write
\(\overline G(s)=\Pr(\|Z\|>s)\). A proposal from \((x,y)\in A\) crossing the cut has length greater than \(S(y)\). Tonelli's theorem and the preceding bound yield
\[
\begin{aligned}
\int_A P(\theta,A^c)\pi_D(d\theta)
&\le \mathbb E_Y\,\overline G(S(Y))\\
&=\mathbb E_Z\!\left[\Pr\{S(Y)<\|Z\|\mid Z\}\right]\\
&\le \frac{C_0}{D}\mathbb E\|Z\|
\le \frac{C_0T_g}{D}.
\end{aligned}
\]
Using \(\pi_D(A)=1/2\), Cheeger's upper bound gives
\(\gamma(P)\le4C_0T_g/D\).

\textit{(2) Uniform lower bound for \(\widetilde P_{H,D}\).}
For sufficiently large \(D\), the boundary-layer condition implies \(U_D(Y_{\mathrm{int}})\ge\zeta_Y>0\), uniformly in \(D\). Let \(\nu_D=\nu_X\otimes U_D\). A teleport into \(Y_{\mathrm{int}}\), followed by \(n_X\) local moves and another teleport, gives
\[
\widetilde P_{H,D}^{\,n_X+2}(\theta,\cdot)
\ge \varepsilon^2\zeta_Y(1-\varepsilon)^{n_X}\eta_X
       \nu_D(\cdot).
\]
The first teleport permits the assumed \(x\)-coordinate minorization; the final teleport supplies the independent \(U_D\) coordinate. Reversibility and Lemma~\ref{lemma:reversible_doeblin_gap} imply the uniform positive gap in part~(2).
\end{proof}

\section{Finite-State Gibbs Calculations}
\label{oa:gibbs_details}
This section gives the matrices and absolute spectral gaps for Example~\ref{ex:toy}, complementing Proposition~\ref{prop:gibbs_gap} by including systematic Gibbs. Use the state order in Example~\ref{ex:toy}, with \(\pi(0,0)=\pi(1,1)=a\) and \(\pi(0,1)=\pi(1,0)=(1-2a)/2\), where \(0<a<1/2\). The single-coordinate updates are
\[
P_1=\begin{pmatrix}
2a&0&1-2a&0\\
0&1-2a&0&2a\\
2a&0&1-2a&0\\
0&1-2a&0&2a
\end{pmatrix},\qquad
P_2=\begin{pmatrix}
2a&1-2a&0&0\\
2a&1-2a&0&0\\
0&0&1-2a&2a\\
0&0&1-2a&2a
\end{pmatrix}.
\]
Set \(P_{\mathrm{RS}}=(P_1+P_2)/2\) and \(P_{\mathrm{sys}}=P_1P_2\). With the teleport kernel \(T\) defined in the continued Example~\ref{ex:toy}, let \(\widetilde P_{H,j}=(1-\varepsilon)P_j+\varepsilon T\) for \(j\in\{\mathrm{RS},\mathrm{sys}\}\) and fixed \(0<\varepsilon<1\). Because systematic Gibbs is generally nonreversible, we compare the finite-state chains using the absolute gap
\[
\gamma_{\mathrm{abs}}(P)
:=1-\max\{|\lambda|:
\lambda\in\operatorname{spec}(P),\ \lambda\ne1\},
\]
where \(\operatorname{spec}(P)\) denotes its eigenvalues. The random-scan kernels here are reversible with nonnegative spectra, so \(\gamma_{\mathrm{abs}}(P)=\gamma(P)\). For systematic Gibbs, the absolute gap is a spectral-radius criterion rather than the reversible variance criterion in Section~\ref{subsec: convergence}. Direct calculation gives
\[
\begin{aligned}
\mathrm{spec}(P_{\mathrm{RS}})&=\{1,2a,1-2a,0\},\\
\mathrm{spec}(P_{\mathrm{sys}})&=\{1,(4a-1)^2,0,0\},\\
\mathrm{spec}(\widetilde P_{H,\mathrm{RS}})
&=\{1,\varepsilon,2a(1-\varepsilon),1-2a(1-\varepsilon)\}.
\end{aligned}
\]
Consequently both local absolute gaps tend to zero as \(a\to1/2\), whereas the hybrid random-scan gap tends to
\(\min\{\varepsilon,1-\varepsilon\}\).
For hybrid systematic Gibbs,
\[
\det(\lambda I-\widetilde P_{H,\mathrm{sys}})
=(\lambda-1)(\lambda-\varepsilon)q_{a,\varepsilon}(\lambda),
\]
where
\[
\begin{aligned}
q_{a,\varepsilon}(\lambda)
&=\lambda^2
-\lambda\{1-16a^2\varepsilon+16a^2+8a\varepsilon-8a\}\\
&\quad-\{8a^2\varepsilon^2-8a^2\varepsilon
-2a\varepsilon^2+2a\varepsilon\}.
\end{aligned}
\]
Its absolute gap is computed from \(1,\varepsilon\) and the two roots of \(q_{a,\varepsilon}\). Figure~\ref{fig:toy_ex1_spectral_gap} compares the four gaps. Both hybrid kernels retain a positive gap as \(a\to1/2\).
\begin{figure}[!htbp]
  \centering
  \includegraphics[width=0.6\textwidth]{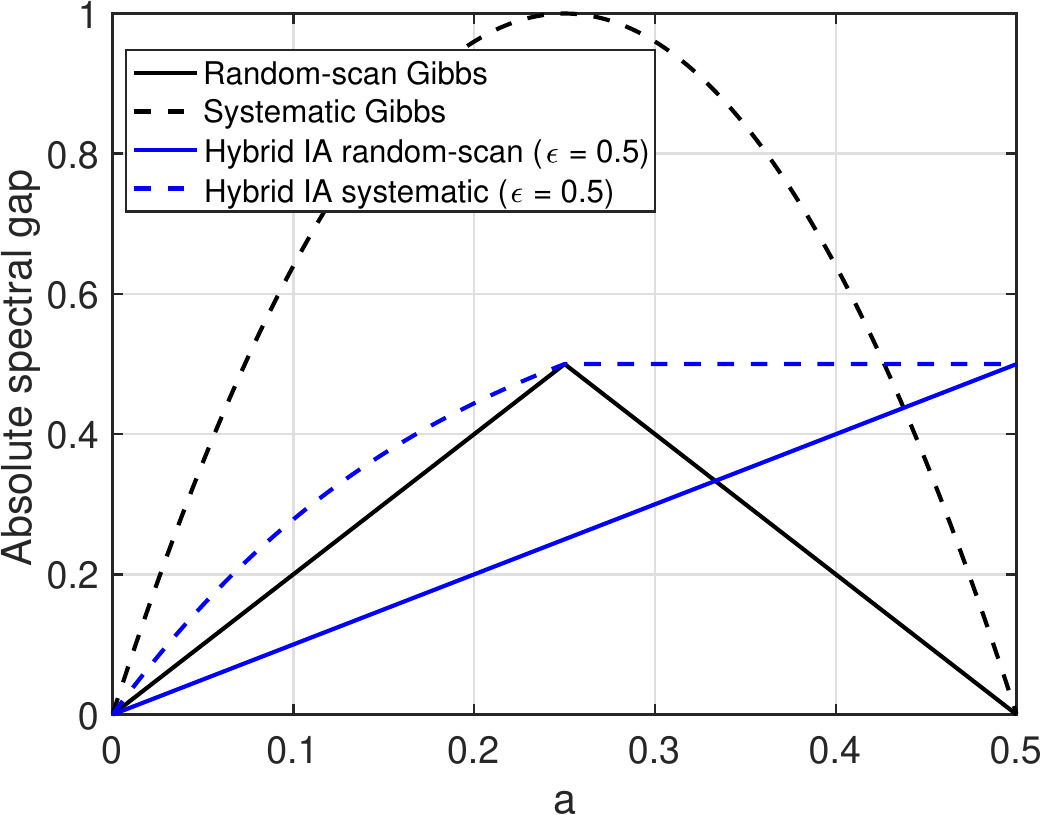}
  \caption{Absolute spectral gaps for random-scan Gibbs, systematic Gibbs, and their hybrid identification-aware versions.}
  \label{fig:toy_ex1_spectral_gap}
\end{figure}

\section{Empirical Implementation Details}
\label{append:empirical}

This section details the teleport updates and diagnostics used in
the SVMA applications in Section~\ref{sec: application}.

\subsection{Prior specification and teleport implementation}

We use the unit-impact normalization, shock-scale convention, and
prior specifications in Section~\ref{sec: application}. All priors
below are densities on the free normalized coefficients and log
shock scales.

Each reported chain contains 100,000 retained transitions, with every
tenth draw stored for the figures and summaries. Within each model and prior,
NUTS and IA--NUTS share the same fixed local tuning.

Identification is most transparent in the frequency domain. Let
\[
H(z)=\sum_{\ell=0}^q\Theta_\ell z^\ell,
\qquad
\Psi(z)=H(z)\operatorname{diag}(\sigma).
\]
Then
\[
f_{yy}(\omega)
=
\Psi(e^{-i\omega})\Psi(e^{-i\omega})^*,
\]
where \(^*\) denotes conjugate transpose. Two parameterizations \(\theta=(\Theta,\sigma)\) and \(\theta'=(\Theta',\sigma')\) are observationally equivalent if and only if they produce the same \(f_{yy}(\omega)\) for every \(\omega\in[-\pi,\pi]\). For parameterizations with nonsingular impact matrices, \citet[Proposition~2]{plagborg2019bayesian}, building on \citet{lippi1994var}, shows that one can pass between any two observationally equivalent parameterizations through a finite sequence of orthogonal rotations and Blaschke root-flip transformations, in an appropriate order, followed by unit-impact normalization. At the transfer-function level, a rotation right-multiplies \(\Psi(z)\) by an orthogonal matrix, while a root flip uses a matrix Blaschke factor \(R(z)\) satisfying
\[
R(e^{-i\omega})R(e^{-i\omega})^*=I_n
\qquad\text{for every }\omega.
\]
Each transformation therefore preserves the spectral density pointwise in frequency. Together, these transformations generate the full observational-equivalence class under the stated nonsingularity condition.

Following \citet[Online Appendix C.4.1]{plagborg2019bayesian}, the initialization routine computes empirical autocovariances through lag \(q\) and runs \(q\) steps of the multivariate innovations algorithm to construct an invertible SVMA\((q)\) starting value. Any root-flip, rotation, or prior-directed search used in this initialization is completed before retained sampling and is not a retained transition.

For the SVMA(1) teleport update, set
\(h(\theta)=(\operatorname{vech}\Gamma_0,
\operatorname{vec}\Gamma_1)\), where \(\Gamma_\ell=\operatorname{Cov}(y_t,y_{t-\ell})\). Here \(\operatorname{vec}\) stacks matrix columns and \(\operatorname{vech}\) stacks the lower triangle of a symmetric matrix. At a given value of \(h\), let \(B\) be the number of root representations considered by the update, \(b\in\{1,\ldots,B\}\) their index, and \(Q\in\mathrm{SO}(3)\) a rotation. Write \(\theta_b^Q\) for the normalized parameters and \(\sigma_{b,j}^Q\) for their shock scales. Relative to counting measure over \(b\) and Haar probability over \(Q\), the conditional target is
\[
\rho_h(b,Q)
\ \propto\
c_b(h)\,
p_{\mathrm{prior}}(\theta_b^Q)
\prod_{j=1}^3(\sigma_{b,j}^Q)^{-6},
\]
where \(6=3(q+1)\) and \(p_{\mathrm{prior}}\) is zero outside the applicable support. With \(\theta_b^\star\) any representative of \(b\), let
\[
c_b(h)
=
\frac{J_{\mathrm{rot}}(\theta_b^\star)}{J_h(\theta_b^\star)}
\prod_{j=1}^3(\sigma_{b,j}^\star)^6.
\]
The coarea Jacobian is
\[
J_h(\theta)=\sqrt{\det\!\left(Dh(\theta)Dh(\theta)^\top\right)}
\] 
at regular points, with derivatives taken in the free \((\Theta,\log\sigma)\) coordinates. The factor \(J_{\mathrm{rot}}(\theta_b^Q)\) is the local volume Jacobian of \(Q\mapsto\theta_b^Q\), from Haar measure to the induced three-dimensional surface measure, up to a common constant. The coarea formula therefore gives a conditional density proportional to \(p_{\mathrm{prior}}J_{\mathrm{rot}}/J_h\).
Right rotations preserve raw-coefficient volume, while normalization has Jacobian \(\prod_j\sigma_j^6\). Consequently,
\(c_b(h)\) is constant within a root representation. It cancels for rotation-only updates but must be retained when
different root representations are mixed.

The proposal density is
\[
q_h(b,Q)=B^{-1}g_{h,b}(Q),
\]
where root representations are proposed uniformly and \(g_{h,b}\) is the rotation proposal within representation \(b\). The proposal is constructed at the given value of \(h\) and therefore does not depend on the current state \((b,Q)\). The common candidate-pool calculation below uses this joint proposal density.

The \(q=16\) teleport update uses rotations only. At a normalized state, set
\[
\Psi_\ell=\Theta_\ell\operatorname{diag}(\sigma).
\]
For \(Q\in\mathrm{SO}(3)\), define
\[
a_j(Q):=(\Psi_0Q)_{jj},
\qquad
D_Q:=\operatorname{diag}\!\bigl(\operatorname{sign}a_1(Q),\ldots,\operatorname{sign}a_3(Q)\bigr).
\]
Whenever \(a_j(Q)\ne0\) for every \(j\), set
\[
\widetilde\Psi_\ell^Q:=\Psi_\ell QD_Q,
\qquad
\sigma_j^Q:=|a_j(Q)|,
\qquad
\Theta_\ell^Q
=
\widetilde\Psi_\ell^Q\operatorname{diag}(\sigma^Q)^{-1}.
\]
This restores unit-impact normalization. Since \(QD_Q\) is orthogonal, the transformation preserves the spectral density and hence the likelihood.

Let \(\theta^Q=(\Theta^Q,\log\sigma^Q)\). Along the rotation orbit through the current state, its target density with respect to Haar probability on \(\mathrm{SO}(3)\) is
\[
\rho(Q)
\ \propto\
p_{\mathrm{prior}}\!\bigl(\theta^Q\bigr)
\prod_{j=1}^3\bigl(\sigma_j^Q\bigr)^{-51},
\]
where \(51=3(q+1)\) and \(p_{\mathrm{prior}}(\theta^Q)\) is defined to be zero outside the applicable parameter bounds. It is constant inside the box under the bounded-uniform prior and is the informative-prior density in the informative specification.

For a given \(q=16\) orbit, the rotation proposal density \(g_{\mathrm{rot}}\) is constructed from its canonical representation and does not depend on the current rotation \(Q\). It has a positive Haar component and hence is positive wherever the displayed conditional target is positive.

The same candidate-pool calculation applies to both applications. Under unit-impact normalization, \(Q\) and \(QH\) yield the same parameters whenever \(H\) is a diagonal sign matrix with determinant one. Both applications use proposal densities invariant under these sign changes and compute pool distances from normalized parameters. The update is therefore independent of the rotation used to represent the current state.

Let \(x\) denote a candidate, \(\rho(x)\) its displayed conditional target density, and \(q(x)\) its proposal density. For \(q=16\), \(x=Q\), there is one representation, and \(q=g_{\mathrm{rot}}\). For SVMA(1), \(x=(b,Q)\) and \(q=q_h\). Suppose \(q(x)>0\) whenever \(\rho(x)>0\). For a pool of size
\(L\), define
\[
\bar\pi(k,x_{1:L})
\ \propto\
\frac1L\rho(x_k)\prod_{\ell\ne k}q(x_\ell),
\qquad
k\in\{1,\ldots,L\}.
\]
Conditional on the pool,
\[
\Pr(k\mid x_{1:L})
=
\frac{\rho(x_k)/q(x_k)}
     {\sum_{\ell=1}^L\rho(x_\ell)/q(x_\ell)}.
\]
Write \(\omega_i=\Pr(i\mid x_{1:L})\). Placing the current candidate in a uniformly chosen slot, drawing the remaining candidates independently from \(q\), and updating the selected index with a kernel reversible with respect to \(\omega\) gives a teleport update that is reversible with respect to the conditional target. Direct selection with probabilities
\(\omega\) is one such update.

The bounded-uniform \(q=16\) teleport update instead favors separated pool
members while retaining the same target. For symmetric
nonnegative distance weights \(H_{ij}=H_{ji}\), with \(H_{ii}=0\),
let
\[
Z_i=\sum_j\omega_jH_{ij}.
\]
From slot \(i\), propose \(j\) with probability
\(\omega_jH_{ij}/Z_i\) and accept with probability
\(\min\{1,Z_i/Z_j\}\). If \(Z_i=0\), remain at \(i\).
For \(i\ne j\), the resulting slot transition satisfies
\[
\omega_i R_{ij}
=
\frac{\omega_i\omega_jH_{ij}}{\max\{Z_i,Z_j\}}
=
\omega_j R_{ji}.
\]
Together with the augmented-pool density above, this establishes
reversibility of the teleport update. A fixed number of
repetitions of the same update preserves this property.
Both constructions use the same conditional-prior and proposal
weights. Local-kernel tuning remains fixed throughout retained
sampling and is not changed after a teleport update.

\subsection{Diagnostic construction and supplementary results}
\label{append:empirical_diagnostics}

For each sampler, we select 100 evenly spaced draws from its 10,000 stored draws and generate 128 conditional completions per selected draw. These completions provide numerical references for the conditional laws defined above. 

For free coefficient \(r\), let \(F_r\) be the empirical CDF of the 100 selected values and \(G_r\) the empirical CDF of the 12,800 pooled conditional completions. The reported coefficient discrepancy is
\[
\overline W_{1,\Theta}
=
\frac{1}{d_\Theta}\sum_{r=1}^{d_\Theta}
\frac{\displaystyle\int_{\mathbb R}
       |F_r(z)-G_r(z)|\,dz}
     {\operatorname{IQR}(G_r)},
\]
where \(\operatorname{IQR}(G_r)\) is the reference interquartile range, \(d_\Theta=15\) for SVMA(1), and \(d_\Theta=150\) for \(q=16\). Fixed impact coefficients and log shock scales are excluded from this average. Near-zero reference IQRs are replaced by the reference standard deviation, and then by one if that is also negligible. Each reference averages over conditioning values visited by its own sampler. It is not an independent full-posterior reference.

The coefficient-rank KS statistic is computed separately for each coefficient from its 100 estimated conditional percentiles and then averaged across coefficients. Density-score KS instead uses conditional percentiles of the conditional log density, with respect to the conditional measures defined above. The conditional log density maps a multidimensional state to a scalar: misallocation among equal-density regions can leave this score unchanged. Coefficient calibration is therefore the primary check for the structural quantities studied here, while density-score KS provides complementary evidence. We report these KS distances descriptively, without independent-sample KS \(p\)-values.

Table~\ref{tab:svma_conditional_diagnostics} shows why the diagnostics should not be treated as interchangeable. In bounded-uniform SVMA(1), rotation-only IA--NUTS improves coefficient calibration but has a larger density-score KS than NUTS. Under the informative priors, the NUTS and IA--NUTS coefficient discrepancies are already small and do not show a consistent improvement. In bounded-uniform \(q=16\), the separately computed log-scale W1 discrepancies are slightly larger for IA--NUTS, despite its better coefficient calibration.

\paragraph{Root and reduced-form checks.}
The root-count discrepancies reported in Section~\ref{subsec:svma1_warmup} are total variation distances \(\frac12\sum_{r=0}^3|\widehat p_r-\overline p_r|\), where \(\widehat p_r\) is the sampled frequency of \(r\) roots outside the unit circle and \(\overline p_r\) is its conditional reference probability averaged over the selected reduced-form states. They assess conditional root-count calibration, not agreement with externally known posterior root weights or mixing among representations with the same count.

Autocovariance comparisons provide a separate check of likelihood-relevant quantities. For \(q=16\), the median coordinatewise KS distance between NUTS and IA--NUTS is \(0.036\) under the bounded-uniform prior and \(0.032\) under the informative prior. Under the informative prior, the corresponding distance between NUTS and IA--RWM is larger, at \(0.095\). In bounded-uniform SVMA(1), the corresponding distance from NUTS is \(0.034\) for rotation-only IA--NUTS but \(0.167\) for IA--NUTS with root flips. Thus improved conditional calibration does not by itself establish full-posterior convergence.

\subsection{Supplemental figures}
Figure~\ref{fig:append_pm_irf_trace_density} reports the informative \(q=16\) coefficient results, using the 10,000 draws stored at intervals of ten retained transitions, with common axes across samplers within each panel.

\begin{figure}[!htbp]
  \centering
  \includegraphics[width=\textwidth]{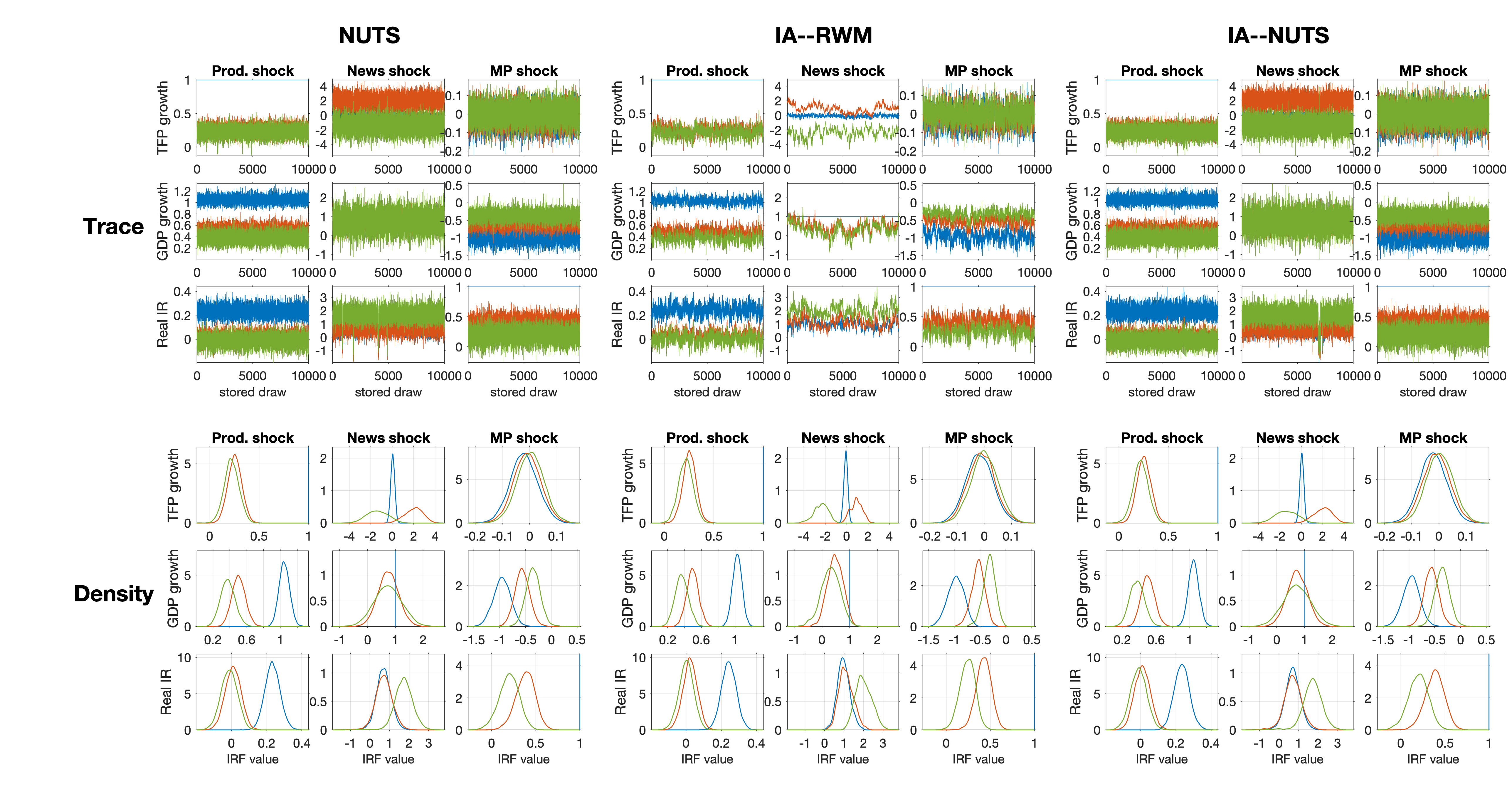}
\caption{Informative-prior \(q=16\) normalized coefficients. Columns correspond to NUTS, IA--RWM, and IA--NUTS. Blue, orange, and green denote lags \(0\), \(1\), and \(2\), respectively.}
  \label{fig:append_pm_irf_trace_density}
\end{figure}

\FloatBarrier

\end{document}